\documentclass[a4paper,10pt]{amsart}
\pdfoutput=1

\usepackage[sc]{mathpazo}
\usepackage{xy}
\usepackage[T1]{fontenc}
\usepackage{subfig}
\usepackage{graphicx}
\usepackage{bm}
\usepackage{amsmath}
\usepackage{amssymb}
\usepackage[active]{srcltx}
\usepackage{hyperref}
\usepackage{enumerate}

\usepackage{color}
\definecolor{red}{rgb}{1,0,0}

\definecolor{gre}{rgb}{0,0.7,0}

\newtheorem{thm}{Theorem}
\newtheorem{lem}{Lemma}
\newtheorem{prop}{Proposition}
\theoremstyle{definition}

\theoremstyle{remark}
\theoremstyle{plain}

\numberwithin{equation}{section}

\def\RR{{\mathbb R}}

\def\ZZ{{\mathbb Z}}

\def\vecH{{\text{\boldmath$H$}}}
\def\veci{{\text{\boldmath$i$}}}

\def\veck{{\text{\boldmath$k$}}}

\def\vecn{{\text{\boldmath$n$}}}
\def\vecq{{\text{\boldmath$q$}}}
\def\vecQ{{\text{\boldmath$Q$}}}
\def\vecp{{\text{\boldmath$p$}}}

\def\vecu{{\text{\boldmath$u$}}}

\def\vecx{{\text{\boldmath$x$}}}

\def\vecy{{\text{\boldmath$y$}}}
\def\vecY{{\text{\boldmath$Y$}}}

\def\vecz{{\text{\boldmath$z$}}}
\def\vecZ{{\text{\boldmath$Z$}}}

\def\vecbeta{{\text{\boldmath$\beta$}}}
\def\vecgamma{{\text{\boldmath$\gamma$}}}

\def\vecnu{{\text{\boldmath$\nu$}}}
\def\veceta{{\text{\boldmath$\eta$}}}

\def\vecxi{{\text{\boldmath$\xi$}}}

\def\scrA{{\mathcal A}}
\def\scrB{{\mathcal B}}

\def\scrQ{{\mathcal Q}}
\def\scrP{{\mathcal P}}

\def\scrS{{\mathcal S}}
\def\scrT{{\mathcal T}}

\def\Re{\operatorname{Re}}

\def\e{\mathrm{e}}
\def\i{\mathrm{i}}
\def\d{\mathrm{d}}
\def\nd{\mathrm{nd}}

\def\C{\operatorname{C{}}}

\def\Op{\operatorname{Op}}

\def\Tr{\operatorname{Tr}}

\def\vol{\operatorname{vol}}

\title[Derivation of the linear Boltzmann equation]{Derivation of the linear Boltzmann equation from the damped quantum Lorentz gas with a general scatterer configuration}

\author{Jory Griffin}
\address{Jory Griffin, School of Mathematics,
University of Bristol,
Bristol BS8 1TW, UK}
\email{\tt j.griffin@bristol.ac.uk}

\thanks{Research supported by EPSRC grant EP/S024948/1}

\begin{document}

\begin{abstract} 
It is a fundamental problem in mathematical physics to derive macroscopic transport equations from microscopic models. In this paper we derive the linear Boltzmann equation in the low-density limit of a damped quantum Lorentz gas for a large class of deterministic and random scatterer configurations. Previously this result was known only for the single-scatterer problem on the flat torus, and for uniformly random scatterer configurations where no damping is required. The damping is critical in establishing convergence  -- in the absence of damping the limiting behaviour depends on the exact configuration under consideration, and indeed, the linear Boltzmann equation is not expected to appear for periodic and other highly ordered configurations.
\end{abstract}

\maketitle


\section{Introduction}
The quantum Lorentz gas is a model of conductivity in which a single quantum particle (electron) evolves in the presence of a potential given by an infinite collection of compactly supported profiles placed on a discrete point set $\scrP \subset \RR^d$. These profiles, called scatterers from here on, represent the relatively heavy molecules of the background material. The point set one should choose, and the limiting behaviour one should expect, is thus dependent on the microscopic structure of the material in question.  A fundamental question is whether one can, for a given $\scrP$, derive a macroscopic transport equation, e.g. the linear Boltzmann equation, from this microscopic model.

Some reasonable choices for $\scrP$ are (i) a realisation of a (Poisson) point process to model disordered materials or an environment with random impurities, (ii) a lattice, union of lattices, or other periodic set to model metals and heavily ordered materials, (iii) aperiodic point sets to model quasicrystals. In the classical (non-quantum) setting, the pioneering papers \cite{Gallavotti69,Spohn78,Boldrighini83} established convergence of the Liouville equation to the linear Boltzmann equation in the low-density (Boltzmann-Grad) limit, provided the scatterer configuration $\scrP$ is random, e.g. given by a homogeneous Poisson point process. More recent work has shown that in the case of crystals \cite{Caglioti10,partII} or other point sets with long-range correlations (e.g. quasicrystals) \cite{MS2019}, different transport equations will emerge in the Boltzmann-Grad limit due to correlations that arise between consecutive collisions. These findings are somewhat mirrored in the quantum setting: On one hand, Eng and Erd\"os \cite{Eng_Erdos} proved convergence to the linear Boltzmann equation for random potentials in the low-density limit, following analogous results in the weak-coupling limit by Spohn \cite{Spohn77} and Erd\"os and Yau \cite{Erdos_Yau}; on the other hand, recent evidence suggests that a different transport law emerges in the same scaling limit when the potential is periodic \cite{GM_SecondOrder,GM_Heuristic}. 

The motivation for the work of the present paper is Castella's striking observation \cite{Castella_LD,Castella_LD2} that the space-homogeneous linear Boltzmann equation can be obtained as the limit of the von Neumann equation on the flat torus with a small scatterer if some damping is introduced. In particular, the evolution for `diagonal' terms is undamped (where incoming and outgoing momenta are equal), and the evolution for `nondiagonal' terms is exponentially damped in time (where incoming and outgoing momenta differ). This exponential damping of nondiagonal terms models phenomenologically the interaction of the system with, for example, a bath of photons or phonons, see \cite{Castella_LD} and references therein, in particular \cite[Chapter 7-3]{SSL}. (Also \cite{Korsch1, Korsch2}). In a rough sense, interactions with a `noisy' external environment can lead to `random' perturbations of the momenta. When the incoming and outgoing momenta are equal, these random perturbations tend to cancel one another out, but when the incoming and outgoing momenta are distinct, these random perturbations persist and lead to exponential decay. Here we will show, using such a damping mechanism, that the \emph{full} (position dependent) linear Boltzmann equation can be obtained as a limit of the quantum Lorentz gas in $\RR^d$ for a general class of scatterer configurations which includes both periodic and disordered examples.

The proof differs from that of the main Theorem in \cite{Castella_LD,Castella_LD2} in a number of ways. If the problem is restricted to the torus one has discrete momenta, and this allows Castella to (i) introduce a damping which is constant on all nondiagonal terms, but zero for diagonal terms, and then (ii) derive a transport equation for the diagonal part of the density matrix before taking any scaling limit to eliminate the nondiagonal terms - the convergence is then established on the level of this transport equation. If one instead considers the problem in $\RR^d$ the momenta are continuous and this approach no longer works. Instead, we (i) introduce a smooth damping function which is zero for diagonal terms and approaches some constant value smoothly as one moves away from the diagonal, and (ii) compute the limit of the full Duhamel expansion, separating damped and undamped regions using a combinatorial argument, and then show that the resulting expression satisfies the linear Boltzmann equation. The damping function in particular must be carefully chosen to scale in the correct way in the small scatterer limit in order to obtain this limiting behaviour, and one must be careful in dealing with the intermediate regime between the undamped and fully damped terms.

We assume in the following that $d\geq 3$. The time evolution of the quantum Lorentz gas is described by the Schr\"odinger equation
\begin{equation}
\frac{i h}{2 \pi} \partial_t \psi(t,\vecx) = H_{h,\lambda} \psi(t,\vecx), 
\end{equation}
where
\begin{equation} \label{Htrunc}
H_{h,\lambda} = -\frac{h^2}{8 \pi^2} \Delta +  \sum_{\vecq \in \scrP} \lambda(r^{d-1} \vecq ) \, W(r^{-1} (\vecx-\vecq)).
\end{equation}
The single-site potential $W$ is assumed to be in the Schwartz class $\scrS(\RR^d)$, $r>0$ is the effective radius of each scatterer, and the $\lambda$ is a cut-off function which we assume to be smooth with compact support contained within the unit ball. The classical mean free path length is $O(r^{1-d})$, so $\lambda$ has the effect of truncating the potential on the macroscopic scale. The assumption that $\lambda$ is compactly supported is a technical one to avoid infinite summation and it's possible that it can be weakened siginificantly. (For example, one may ideally wish to take $\lambda(\vecq)$ constant.)

We assume that $\scrP\subset\RR^d$ is a uniformly discrete point set with asymptotic density one. This technical requirement is introduced so that $\scrP$ provides a suitable set over which a $d$-dimensional Riemann sum can be computed, and that this Riemann sum converges with an explicit error term. In particular, we require that there exists $b_\scrP,c_\scrP>0$ such that $\|\vecq-\vecq'\| > b_\scrP$ for all $\vecq,\vecq'\in\scrP$ with $\vecq \neq \vecq'$ and for every $g\in\C_c^\infty(\RR^d)$, $0<\epsilon<1$ we have
\begin{equation} \label{assumption1}
\epsilon^d \sum_{\vecq \in \scrP} g\left( \epsilon\vecq\right) = \int_{\RR^d} g(\vecx) \d \vecx + O(\epsilon^{ c_\scrP} \| \nabla g \|) .
\end{equation}
Deterministic examples of $\scrP$ that satisfy these assumptions are lattices (e.g. $\scrP=\ZZ^d$) and large classes of quasicrystals (e.g. the vertices of a Penrose tiling). For random examples one can take the so-called Mat\'ern processes \cite{Matern} in which a realisation of a homogeneous Poisson point process is then \emph{thinned} to remove clusters, or a random displacement model, in which each point in a deterministic set (e.g. a lattice) is randomly perturbed by a small amount. (As long as the random perturbation is small enough the resulting point set will be uniformly discrete provided the initial point set is uniformly discrete). The restriction to uniform discreteness likely can be weakened. For example, one may wish to take $b_\scrP$ to depend on $\|\vecq\|$ and $\|\vecq'\|$, or insist that that $\|\vecq-\vecq'\|>b_\scrP$ holds only for almost all pairs of points in $\scrP$. In both cases we expect the same results to hold.

To study the quantum transport and the Boltzmann-Grad limit, it is convenient to move to the equivalent Heisenberg picture and study the quantum Liouville equation (or von Neumann equation/ backward Heisenberg equation)
\begin{equation} \label{heisenberg}
\partial_t \rho_t = - \frac{2\pi \i}{h} [H_{h,\lambda}, \rho_t]
\end{equation}
for a density operator $\rho_t$. We introduce damping to the system by considering the $\alpha$-damped von Neumann equation (in momentum representation):
\begin{equation} \label{dampedheisenberg}
\partial_t \, \widehat\rho_t(\vecy,\vecy') = - \frac{2\pi \i}{h} [\widehat H_{h,\lambda}, \widehat\rho_t](\vecy,\vecy') - \frac{\alpha^d}{h} \, \big(1-\Gamma(\alpha h^{1-d}(\vecy-\vecy'))\big) \, \widehat\rho_t(\vecy,\vecy'),
\end{equation}
where $\alpha \geq 0$ is the strength of the damping and $\Gamma\in\C_c^\infty(\RR^d)$ with values in $[0,1]$ so that $\Gamma(\vecy)=1$ in some neighbourhood of the origin and $\Gamma(\vecy) = 0$ for $\|\vecy\|>1$. Eq.~\eqref{dampedheisenberg} describes the averaged quantum dynamics of a particle subject to white noise in momentum where $\Gamma(\vecy)$ is the covariance function of the corresponding Gaussian random field. We refer the reader to \cite{Fischer} for detailed rigorous treatment of white noise perturbations in phase space, and to \cite{Hislop,Jayannavar_Kumar, Madhukar_Post} for the more standard setting in position space.

In order to establish the convergence of the damped von Neumann equation \eqref{dampedheisenberg} to the linear Boltzmann equation, we need to carefully prepare the initial condition of $\rho_t$ relative to a classical phase space density $a$. Following the approach in \cite{GM_SecondOrder}, we achieve this by the rescaled Weyl quantisation $\Op_{r,h}(a)$ of a classical phase-space symbol $a$:
\begin{equation}\label{def:Opa}
\Op_{r,h}(a) f(\vecx) = r^{d(d-1)/2} h^{d/2} \int_{\RR^{2d}} a(\tfrac12 r^{d-1}  (\vecx+\vecx'),h \vecy) \, \e((\vecx-\vecx')\cdot\vecy)\, f(\vecx')\, \d \vecx' \d \vecy ,
\end{equation}
with the shorthand $\e(x) := \e^{2 \pi \i x}$. This means we measure momenta on the semi-classical scale, and position on the scale of the classical mean free path. Although other scalings are possible, we will here focus on the case when $r=h$. This will ensure that scattering remains truly quantum in the limit $r\to 0$, and that we see the full quantum $T$-operator in the limit. For the single scatterer Hamiltonian 
\begin{equation} \label{Htrunc}
H_{\mu} = -\frac{1}{8 \pi^2} \Delta + \mu \, W(\vecx),
\end{equation}
we define the $T$-operator at energy $E$ to be the operator satisfying
\begin{equation}\label{def:TE}
T_\mu(E)=\mu \,  \Op_{1,1}(W)  \left( 1 + \frac{1}{E- H_0 +\i 0_+} \, T_\mu(E) \right)
\end{equation}
and write $T_\mu(\vecy,\vecy')$ for its integral kernel in momentum representation at energy $E=\tfrac12 \|\vecy\|^2$. We have the explicit expansion (understood in terms of distributions)
\begin{equation}
\begin{split} \label{Texplicit}
T_\mu(\vecy,\vecy') &= \mu \widehat W(\vecy-\vecy') \\
&+ \sum_{\ell=1}^{\infty}(-2\pi \i)^{\ell} \mu^{\ell+1} \int_{\RR^{d \ell} } \widehat W(\vecy -\vecy_1) \cdots \widehat W(\vecy_\ell-\vecy') \\
&\hspace{3cm} \times  [\prod_{i=1}^{\ell}\int_0^{\infty} \e(\tfrac12 ( \|\vecy\|^2- \|\vecy_i\|^2) u ) \, \d u ]  \, \d \vecy_1 \cdots \d \vecy_\ell
\end{split}
\end{equation}
where
\begin{equation}
\widehat W(\vecy) := \int_{\RR^d} W(\vecx) \, \e(-\vecx\cdot\vecy) \, \d \vecx.
\end{equation}

\begin{thm} \label{theorem}
Let $a,b$ be in the Schwartz class $\scrS(\RR^d \times \RR^d)$. If $\rho_t$ is a solution of the $\alpha$-damped von Neumann equation \eqref{dampedheisenberg} subject to the initial condition $\rho_0 = \Op_{r,h}(a)$, then for $t>0$
\begin{equation}
\lim_{\alpha\to 0}\lim_{r=h\to 0} \Tr(\rho_{r^{1-d} t} \Op_{r,h}(b)) = \int_{\RR^{2d}} f(t,\vecx,\vecy) b(\vecx,\vecy) \d \vecx \d \vecy
\end{equation}
where $f(t,\vecx,\vecy)$ solves the linear Boltzmann equation 
\begin{equation}\label{LB2}
\begin{cases} 
\displaystyle
\bigl(\partial_t+\vecy\cdot\nabla_\vecx\bigr) f(t,\vecx,\vecy) 
=
\int_{\RR^d}  \big[ \Sigma_{\lambda(\vecx)}(\vecy,\vecy') f(t,\vecx,\vecy') -  \Sigma_{\lambda(\vecx)}(\vecy',\vecy)  f(t,\vecx,\vecy) \big] \,\d\vecy' & \\
 f(0,\vecx,\vecy)=a(\vecx,\vecy) & 
 \end{cases}
\end{equation}
with the collision kernel
\begin{equation} \label{LBEkernel}
\Sigma_{\mu}(\vecy,\vecy') = 8 \pi^2 | T_{\mu}(\vecy,\vecy')|^2 \, \delta(\|\vecy\|^2-\|\vecy'\|^2).
\end{equation}
\end{thm}

Note that the limits $\alpha \to 0$ and $r \to 0$ do not commute. Indeed if one first takes the limit $\alpha \to 0$ followed by $r \to 0$ one is back in the situation of \cite{Eng_Erdos,GM_SecondOrder,GM_Heuristic} where the limit depends on the precise nature of $\scrP$. The striking feature of Theorem \ref{theorem} is that the limit is the same for all admissible scatterer configurations $\scrP$, from periodic to highly disordered.

In Section 2 we perform the Duhamel expansion of the solution to the damped Heisenberg equation, this allows us to obtain an explicit formal expansion for the solution as a power series in $\lambda(\vecx)$. In Section 3 we perform a carefully chosen partition of unity which allows us to isolate the damped and undamped regions. In Section 4 we perform the low-density followed by the zero damping limit on this reorganised series. This Section constitutes the bulk of the paper: we first show that the sum of all nondiagonal terms converges, and then vanishes in the limit; then we show that the sum of all diagonal terms converges, and hence that the entire series converges to some $f(t,\vecx,\vecy)$ given explicitly as an expansion in $\lambda$. In Section 5 we prove that our limiting expression coincides with a solution of the linear Boltzmann equation using \cite{Castella_LD2}.

\section{Deriving a Formal Expansion}

In the momentum representation, the kernel of the Hamiltonian \eqref{Htrunc} reads 
\begin{equation}
\widehat{H}_{h,\lambda}(\vecy,\vecy') = \frac{h^2}{2} \|\vecy\|^2 \, \delta(\vecy-\vecy') + \widehat{\Op}(V)(\vecy,\vecy')
\end{equation}
where
\begin{equation}
\widehat{\Op}(V)(\vecy,\vecy') =r^d \sum_{\vecq \in \scrP}  \lambda(r^{d-1} \vecq )  \e(\vecq\cdot(\vecy'-\vecy)) \widehat W(r(\vecy-\vecy')).
\end{equation}
Inserting these into \eqref{dampedheisenberg} yields, after a suitable variable substitution,
\begin{equation}\begin{split}\label{EE0}
\partial_t \widehat\rho_t(\vecy,\vecy') = & -\left( \pi \i \, h (  \|\vecy\|^2 - \|\vecy'\|^2) +  \frac{\alpha^d}{h} (1-\Gamma(\alpha h^{1-d}(\vecy-\vecy'))) \right) \widehat\rho_t(\vecy,\vecy') \\
& - \frac{2 \pi \i}{h}  r^d \sum_{\vecq \in \scrP} \lambda (r^{d-1} \vecq )  \int_{\RR^d} \d \vecz \, \e( -\vecq\cdot \vecz) \widehat W(r\vecz) \, [\widehat\rho_t(\vecy-\vecz,\vecy') -\widehat\rho_t(\vecy,\vecy'+\vecz) ].
\end{split}\end{equation}
Following Castella \cite{Castella_LD}, it will be convenient to write 
\begin{equation}
\widehat\rho_t(\vecy-\vecz,\vecy') -\widehat\rho_t(\vecy,\vecy'+\vecz)=  - \sum_{\gamma\in\{0,1\}} (-1)^\gamma \widehat\rho_t(\vecy- \gamma \vecz,\vecy'+\bar\gamma \vecz) 
\end{equation}
with $\bar\gamma:=1-\gamma$.
The Duhamel principle for \eqref{EE0} yields 
\begin{equation}\begin{split}
\widehat\rho_t(\vecy,\vecy') &=  \e( - \tfrac{h}{2} (  \|\vecy\|^2 - \|\vecy'\|^2) \, t ) \, \e^{ -\frac{\alpha^d}{h} (1-\Gamma(\alpha h^{1-d}(\vecy'-\vecy))) \, t } \widehat\rho_0(\vecy,\vecy') \\
& + \frac{2 \pi \i}{h}  r^d \sum_{\vecq \in \scrP} \lambda(r^{d-1} \vecq )  \int_{\RR^d} \d \vecz \, \e( -\vecq\cdot \vecz) \widehat W(r\vecz)  \sum_{\gamma\in\{0,1\}} (-1)^\gamma  \\
& \times \int_0^t \e( - \tfrac{h}{2} (  \|\vecy\|^2 - \|\vecy'\|^2) \, (t-s) ) \, \e^{ - \frac{\alpha^d}{h} (1-\Gamma(\alpha h^{1-d}(\vecy'-\vecy))) \, (t-s) }  \widehat\rho_s(\vecy- \gamma \vecz,\vecy'+\bar\gamma \vecz) \, \d s.
\end{split}\end{equation}
Iterating this expression and making the substitutions $u_0 = t-s_1$ and $u_j = s_{j}-s_{j+1}$ for $j\geq 1 $ we obtain the formal expansion
\begin{equation}\begin{split} \label{expansionA0}
\widehat\rho_t(\vecy,\vecy') &= \e( - \tfrac{h}{2} (  \|\vecy\|^2 - \|\vecy'\|^2) \, t ) \, \e^{ -\frac{\alpha^d}{h} (1- \Gamma(\alpha h^{1-d}(\vecy'-\vecy))) \, t } \widehat\rho_0(\vecy,\vecy') \\
& + \sum_{m=1}^\infty (2 \pi \i h^{-1} r^d)^m \sum_{\vecq_1,\cdots,\vecq_m \in \scrP} \lambda(r^{d-1} \vecq_1 ) \cdots \lambda(r^{d-1} \vecq_m) \\
& \times \int_{\RR^{md}} \d \vecz_1\cdots\d\vecz_m \e( -\vecq_1\cdot \vecz_1-\cdots - \vecq_m\cdot\vecz_m) \,  \widehat W(r\vecz_1) \cdots \widehat W(r\vecz_m) \\
 & \times\sum_{\gamma_1, \cdots, \gamma_m \in\{0,1\}} (-1)^{\gamma_1+\cdots+\gamma_m}  
\int_{
\triangle_m(t)} \d u_0\cdots \d u_{m}  \\
& \times \left[ \prod_{j=0}^{m}\e( - \tfrac{h}{2} (  \|\vecy -  \sum_{i=1}^j \gamma_i \vecz_i \|^2  - \|\vecy' + \sum_{i=1}^j \bar\gamma_i \vecz_i \|^2) \, u_j ) \right] \\
& \times \left[ \prod_{j=0}^{m} \e^{ - \frac{\alpha^d}{h} (1-\Gamma(\alpha h^{1-d}(\vecy' - \vecy + \sum_{i=1}^j \vecz_i)))  \, u_j } \right] \widehat\rho_0(\vecy -  \sum_{i=1}^m \gamma_i \vecz_i, \vecy' + \sum_{i=1}^m \bar\gamma_i \vecz_i)
\end{split}\end{equation}
where $\triangle_m(t) \subset \RR^{m+1}$ is the set 
$$ \triangle_m(t) = \{  (u_0,\dots,u_m) \in \RR_+^{m+1} \mid u_0+\cdots+u_{m} = t \}.$$
We now wish to compute $\Tr(\rho_{r^{1-d} t} \Op_{r,r}(b))=\Tr(\widehat\rho_{r^{1-d} t} \widehat\Op_{r,r}(b))$, where $\rho_t$ solves the damped von Neumann equation \eqref{dampedheisenberg} with initial condition $\rho_0 = \Op_{r,r}(a)$. 
The kernel of $\Op_{r,h}(a)$ as defined in \eqref{def:Opa} reads in momentum representation
\begin{equation}
\widehat \Op_{r,h}(a) (\vecy,\vecy') = r^{-d(d-1)/2}h^{d/2}\tilde a(r^{1-d}(\vecy-\vecy'), \tfrac{h}{2}(\vecy+\vecy') )
\end{equation}
where $\tilde a(\vecxi,\vecy) = \int_{\RR^d}a(\vecx,\vecy) \e(-\vecx\cdot\vecxi) \d \vecx$. Inserting these in \eqref{expansionA0} yields the expansion
\begin{equation}
\Tr (\rho_{r^{1-d} t} \Op_{r,r}(b) ) = \sum_{m=0}^\infty (2 \pi \i)^m \, \scrA_m^{\alpha,r}(t)
\end{equation}
where
\begin{multline}
\scrA_0^{\alpha,r}(t) = r^{-d(d-2)} \int_{\RR^{2d}} \d \vecy \d \veceta \, \e( - \tfrac{1}{2} r^{2-d}  (  \|\vecy\|^2 - \|\veceta\|^2) \, t )\e^{ - \alpha^d r^{-d} (1-\Gamma(\alpha r^{1-d}(\veceta-\vecy)))t }  \\ \times \tilde a(r^{1-d}(\vecy-\veceta), \tfrac{r}{2}  (\vecy+\veceta) ) \, \tilde b(r^{1-d}(\veceta-\vecy), \tfrac{r}{2}(\veceta+\vecy)) ,
\end{multline}
and for $m\geq 1$
\begin{equation}
\begin{split}
\scrA_m^{\alpha,r}(t) &= r^{(m-d)(d-1)+d}\sum_{\vecq_1,\cdots,\vecq_m \in \scrP} \lambda(r^{d-1} \vecq_1 ) \cdots \lambda(r^{d-1} \vecq_m) \\
&\times \sum_{\gamma_1, \cdots, \gamma_m \in\{0,1\}} (-1)^{\gamma_1+\cdots+\gamma_m} \int_{\RR^{2d}} \d \vecy \d \veceta \int_{\RR^{md}} \d \vecz_1\cdots\d\vecz_m \\
&\times \left[ \prod_{i=1}^m \e(-\vecq_i \cdot \vecz_i) \, \widehat W(r\vecz_i) \right]  \int_{\triangle_m(r^{1-d} t)} \d u_0\cdots \d u_{m}\, \left[ \prod_{j=0}^{m} \e^{ -\frac{\alpha^d}{r} (1-\Gamma(\alpha r^{1-d}(\veceta - \vecy + \sum_{i=1}^j \vecz_i)) )\, u_j } \right] \\
&\times\left[ \prod_{j=0}^{m}\e( \tfrac{r}{2} (  \|\veceta +  \sum_{i=1}^j \bar\gamma_i \vecz_i \|^2  - \|\vecy - \sum_{i=1}^j \gamma_i \vecz_i \|^2) \, u_j ) \right]\\
&\times \tilde a(r^{1-d}(\vecy-\veceta-\sum_{i=1}^m \vecz_i), \tfrac{r}{2}(\vecy+\veceta - \sum_{i=1}^m (\gamma_i-\bar\gamma_i)\vecz_i) ) \, \tilde b(r^{1-d}(\veceta-\vecy), \tfrac{r}{2}(\veceta+\vecy) ).
\end{split}
\end{equation}
We first make the substitution $\veceta \to \vecy + r^{d-1} \veceta$. Then, make the substitution $\vecy \to r^{-1}\vecy$ and for all $j$, make the substitutions $u_j \to r u_j$, $\vecz_j \to r^{-1} \vecz_j$. This yields the expression
\begin{equation} \label{simple}
\begin{split}
\scrA_m^{\alpha,r}(t) &=  \sum_{\vecq_1,\cdots,\vecq_m \in \scrP} \lambda(r^{d-1} \vecq_1 ) \cdots \lambda(r^{d-1} \vecq_m)\, \sum_{\gamma_1, \cdots, \gamma_m \in\{0,1\}} (-1)^{\gamma_1+\cdots+\gamma_m}   \\ 
&\times \int_{\RR^{2d}}\d\vecy\d\veceta \int_{\RR^{md}} \d \vecz_1\cdots\d\vecz_m \left[ \prod_{i=1}^m \e(- r^{-1} \vecq_i \cdot \vecz_i) \, \widehat W(\vecz_i) \right] \\
&\times \int_{\triangle_m(r^{-d}t)} \d u_0\cdots \d u_{m} \,  \left[ \prod_{j=0}^{m}\e( \xi_j \, u_j ) \, \e^{- \alpha^d (1-\Gamma(\alpha (\veceta+r^{-d}\sum_{i=1}^j \vecz_i) )) \, u_j} \right]\\
&\times \tilde a(- \veceta - r^{-d} \sum_{i=1}^m \vecz_i,  \vecy - \sum_{i=1}^m \gamma_i \vecz_i + \tfrac{1}{2} r^{d}\veceta + \tfrac12 \sum_{i=1}^m \vecz_i) ) \, \tilde b(\veceta, \vecy + \tfrac12 r^{d} \veceta)
\end{split}
\end{equation}
where $\xi_j$ is given by
\begin{equation}\begin{split}
\xi_j &= \tfrac12(\|\vecy +  \sum_{i=1}^j \bar\gamma_i \vecz_i +r^{d}\veceta\|^2)- \|\vecy -  \sum_{i=1}^j \gamma_i \vecz_i \|^2 ) \\
&= (\vecy- \sum_{i=1}^j \gamma_i \vecz_i) \cdot (\sum_{i=1}^j \vecz_i + r^d \veceta ) +  \tfrac12 \| \sum_{i=1}^j \vecz_i + r^d \veceta \|^2.
\end{split}\end{equation}
The limit of the first term can be computed immediately.
\begin{prop} \label{prop:mequalszero}
\begin{equation}
\lim_{\alpha \to 0} \lim_{r\to 0} \scrA_0^{\alpha,r}(t) = \int_{\RR^{2d}} \d \vecx \d \vecy \, a(\vecx-t\vecy, \vecy) \, b(\vecx,\vecy).
\end{equation}
\end{prop}
\begin{proof}
We have that
\begin{equation}
\scrA_0^{\alpha,r}(t) =  \int_{\RR^{2d}} \d \vecy \d \veceta \, \e( \vecy\cdot\veceta t + r^{d}\|\veceta\|^2 \, t ) \\ 
 \e^{ - r^{-d}\alpha^d (1-\Gamma(\alpha \veceta)) \, t } \tilde a(-\veceta, \vecy + \tfrac12 r^d \veceta )\, \tilde b(\veceta,  \vecy + \tfrac12 r^d \veceta).
\end{equation}
The functions $\tilde a$ and $\tilde b$ are rapidly decaying so this is uniformly bounded as $r \to 0$. By dominated convergence we thus obtain
\begin{equation}
\lim_{r\to 0} \scrA_0^{\alpha,r}(t) = \int_{\RR^{2d}} \d \vecy \d \veceta \, \e( \vecy\cdot\veceta t )  \tilde a(-\veceta, \vecy )\, \tilde b(\veceta,  \vecy ) \, \bm 1 [\Gamma(\alpha \veceta) = 1 ].
\end{equation}
Again, by the rapid decay of $\tilde a$ and $\tilde b$ this converges in the limit $\alpha \to 0$ and we obtain
\begin{equation}
\lim_{\alpha \to 0}\lim_{r\to 0} \scrA_0^{\alpha,r}(t) =\int_{\RR^{2d}} \d \vecy\d \veceta \, \e( \vecy\cdot\veceta t )  \tilde a(-\veceta, \vecy )\, \tilde b(\veceta,  \vecy) \\
= \int_{\RR^{2d}} \d \vecx \d \vecy \, a(\vecx-t\vecy, \vecy) \, b(\vecx,\vecy).
\end{equation}
\end{proof}

\section{Manipulating the Expansion}
For the higher order terms we perform a partitioning of the $\vecz_i$ integration region. To see why, note that \eqref{simple} has a product of factors of the form
\begin{equation}
\e^{-\alpha^d (1-\Gamma(\alpha(\veceta+r^{-d} \sum_{i=1}^j \vecz_i))) \, u_j }.
\end{equation}
If the argument $\alpha(\veceta+r^{-d} \sum_{i=1}^j \vecz_i)$ is large, then this entire factor becomes $\e^{-\alpha^d \, u_j}$, and hence the $u_j$ integral is exponentially damped. Our partition will be precisely into these damped and undamped regions. Let $\scrS = \{ s_1,\cdots,s_{p}\}\subset \{0,\cdots,m\}$ with $s_1 = 0$ and $s_{p}=m$ and write $\Pi_m$ for the set of all such $\scrS$. Define $\chi^\scrS : \RR^{d(m-1)} \to \RR$ by
\begin{equation}
 \chi^{\scrS}(\vecz_1,\dots, \vecz_{m-1}) = \left[ \prod_{j=2}^{p-1} \chi( \vecz_{s_j}) \right] \left[ \prod_{j \notin \scrS} (I-\chi)(\vecz_j) \right]
 \end{equation}
 where $\chi \in C_c^\infty(\RR^d\to \RR)$ is decreasing in $\|\vecz\|$ such that $\chi(\vecz)=1$ for all $\|\vecz\|<1$ and $\chi(\vecz) = 0$ for all $\|\vecz\| > 2$. This implies the bound
 \begin{equation} \label{chibound}
 \| \chi \|_{L^1} \leq \vol ( \scrB_{2})
 \end{equation}
 where $\scrB_r$ is the $d$-ball of radius $r$. Note that $\chi^{\scrS}$ forms a partition of unity: $\sum_{\scrS \in \Pi_m} \chi^{\scrS} = 1$; and also that by assumption on the support of $\Gamma$
\begin{equation}
(1-\chi(\alpha \vecz)) \e^{- \alpha^d (1-\Gamma(\alpha \vecz)) u} = (1-\chi(\alpha \vecz)) \e^{- \alpha^d u}.
\end{equation}
We put $\vecgamma = (\gamma_1,\dots,\gamma_m)$ and rewrite \eqref{simple} as
\begin{equation} \label{scrAseries}
\scrA_m^{\alpha,r}(t)=   \sum_{\vecgamma \in\{0,1\}^m} (-1)^{\gamma_1+\cdots+\gamma_m}  \sum_{\scrS\in \Pi_m}\scrA_{\vecgamma,\scrS}^{\alpha,r}(t) 
\end{equation}
where
\begin{equation} \label{simple2}
\begin{split}
\scrA_{\vecgamma,\scrS}^{\alpha,r}(t)  &= \sum_{\vecq_1, \cdots,\vecq_m \in \scrP} \lambda(r^{d-1} \vecq_1 ) \cdots \lambda(r^{d-1} \vecq_m)  \int_{\RR^{(m+2)d}}\d\vecy\d\veceta \d\vecz_1 \cdots \d\vecz_m \\
&\times \widehat W(\vecz_1) \cdots \widehat W(\vecz_m) \, \e(-r^{-1}\vecq_1 \cdot \vecz_1-\cdots-r^{-1}\vecq_m \cdot\vecz_m)\\
&\times\int_{\triangle_m(r^{-d}t)} \d u_0\cdots \d u_{m} \,  \left[ \prod_{j=0}^{m}\e( \xi_j u_j ) \e^{-\alpha^d (1-\Gamma(\alpha(\veceta+r^{-d}\sum_{i=1}^j \vecz_i)) )u_j} \right]\\
&\times \chi^{\scrS}(\alpha (\veceta+r^{-d} \vecz_1),\dots,\alpha (\veceta+r^{-d}(\vecz_1+\cdots+\vecz_{m-1}) )) \\ 
&\times \tilde a(- \veceta - r^{-d} \sum_{i=1}^m \vecz_i,  \vecy - \sum_{i=1}^m \gamma_i \vecz_i  +  \tfrac{1}{2}( r^{d}\veceta + \sum_{i=1}^m \vecz_i) ) \, \tilde b(\veceta, \vecy + \tfrac12 r^{d} \veceta).
\end{split}
\end{equation}
Note that all elements in the complement of $\scrS$ occur in $|\scrS|-1 = p-1$ contiguous blocks (possibly of size zero). Write $\kappa_i = s_{i+1}-s_i-1\geq 0$ for the number of elements in the $i^{\text{th}}$ block. To simplify notation we will use double subscripts to refer to the $j$th element of the $i$th block, e.g. $\vecz_{ij}:= \vecz_{s_i+j}$ where $0\leq j \leq \kappa_i$. When $j=0$ we will write $\vecz_{s_i}$ or $\vecz_{i0}$ interchangeably. We then put $\veceta=\veceta_1$ and for $i=2,\cdots, p$ we make the change of coordinates for $\vecz_{(i+1)0} $ by
$$ \veceta_{i+1}= r^{-d} (\vecz_{(i+1)0} + \sum_{j=1}^{\kappa_{i}} \vecz_{ij}).$$
This gives a factor of $r^{d^2(p-1)}$. In these new coordinates we have that
\begin{equation}
\vecq_1\cdot \vecz_1+\cdots + \vecq_m\cdot\vecz_m = \sum_{i=1}^{p-1}( r^d \vecq_{(i+1)0}\cdot \veceta_{i+1} + \sum_{j=1}^{\kappa_i} (\vecq_{ij} - \vecq_{(i+1)0}) \cdot\vecz_{ij} ).
\end{equation}
The product of potentials can be written
\begin{equation}
\widehat W(\vecz_1) \cdots \widehat W(\vecz_m) = \prod_{i=1}^{p-1} \widehat W(r^d \veceta_{i+1} - \sum_{j=1}^{\kappa_i} \vecz_{ij} )\prod_{j=1}^{\kappa_i} \widehat W(\vecz_{ij}) .
\end{equation}
By convention let us assume that $\gamma_{s_1}=\gamma_0 = 0$. For $i = 1,\cdots,p$ we have that $\xi_{s_i} = r^d \zeta_i$ where
\begin{equation}
\zeta_i  =  (\vecy -  \sum_{k=1}^{i-1} \sum_{\ell=1}^{\kappa_k} (\gamma_{k\ell} - \gamma_{(k+1)0} ) \vecz_{k\ell} )\cdot (\sum_{k=1}^i \veceta_k) \\ 
+ \tfrac12 r^d ( \| \sum_{k=1}^{i} \bar \gamma_{s_k} \veceta_k \|^2 -\| \sum_{k=1}^{i} \gamma_{s_k} \veceta_k \|^2).
\end{equation}
For $i=1,\dots,p-1$ and $j = 1, \dots, \kappa_i$ we have that
\begin{multline}
\xi_{ij} = \tfrac12(\|\vecy +  \sum_{k=1}^{i-1}\sum_{\ell=1}^{\kappa_k} (\bar\gamma_{k\ell} - \bar\gamma_{(k+1)0}) \vecz_{k\ell} + \sum_{\ell=1}^j \bar\gamma_{i\ell} \vecz_{i\ell} +r^{d} \sum_{k=1}^i \bar\gamma_{k0} \veceta_k\|^2 \\
 - \|\vecy -  \sum_{k=1}^{i-1}\sum_{\ell=1}^{\kappa_k} (\gamma_{k\ell} - \gamma_{(k+1)0}) \vecz_{k\ell}- \sum_{\ell=1}^j \gamma_{i\ell} \vecz_{i\ell} -r^{d} \sum_{k=1}^i \gamma_{k0} \veceta_k\|^2 ).
\end{multline}
The functions $\tilde a$ and $\chi^{\scrS}$ become
\begin{multline}
 \tilde a(- \veceta - r^{-d} \sum_{i=1}^m \vecz_i,  \vecy - \sum_{i=1}^m \gamma_i \vecz_i  +  \tfrac{1}{2}( r^{d}\veceta + \sum_{i=1}^m \vecz_i) ) \\
= \tilde a(- \sum_{i=1}^{p} \veceta_i,  \vecy - \sum_{i=1}^{p-1} \sum_{j=1}^{\kappa_i} (\gamma_{ij}-\gamma_{(i+1)0})  \vecz_{ij}  - \tfrac12 r^d \sum_{i=1}^p (\gamma_{i0}-\bar \gamma_{i0})\veceta_i ),
\end{multline}
and
\begin{multline}
 \chi^{\scrS}(\alpha (\veceta+r^{-d} \vecz_1),  \dots,\alpha (\veceta+r^{-d}(\vecz_1+\cdots+\vecz_{m-1}))) \\
= \left(\prod_{i=2}^{p-1} \chi( \alpha \sum_{k=1}^i \veceta_k) \right) \left( \prod_{i=1}^{p-1} \prod_{j=1}^{\kappa_i} (I-\chi)( \alpha (\sum_{k=1}^i \veceta_k + r^{-d} \sum_{\ell=1}^j \vecz_{i\ell} ) ) \right).
\end{multline}
We write $\vecH = (\veceta_1,\cdots,\veceta_{p}) $ and 
$$
\vecZ_\scrS=(\vecz_{11},\ldots,\vecz_{1\kappa_1},  \ldots, \vecz_{(p-1)1},\ldots,\vecz_{(p-1)\kappa_{p-1}})
$$ 
for the collection of remaining $\vecz_i$ variables. Make the substitution $u_{s_i} = r^{-d}\nu_i$, then equation \eqref{simple2} can now be written
\begin{equation} \label{simple3}
\begin{split}
\scrA_{\vecgamma,\scrS}^{\alpha,r}(t)  &= r^{d(d-1)(p-1)}\sum_{\vecq_1, \cdots,\vecq_m \in \scrP} \lambda(r^{d-1} \vecq_1 ) \cdots \lambda(r^{d-1} \vecq_m) \int_{\RR^{(m+2)d}}\d\vecy \d \vecH \d\vecZ_\scrS\,  \\
&\times F_{\vecgamma,\scrS}^r(\vecZ_\scrS,\vecH,\vecy) \left[ \prod_{i=1}^{p-1} \prod_{j=1}^{\kappa_i} (I-\chi)(\alpha ( \sum_{k=1}^i \veceta_k + r^{-d} \sum_{\ell=1}^j \vecz_{i\ell})) \right] \\
&\times  \left[ \prod_{i=1}^{p-1} \e\left(-r^{d-1} \vecq_{(i+1)0} \cdot\veceta_{i+1} - r^{-1} \sum_{j=1}^{\kappa_i} (\vecq_{ij} - \vecq_{(i+1)0}) \cdot\vecz_{ij}  \right) \right] \\
&\times\left[\prod_{i=2}^{p-1}\chi(\alpha\sum_{k=1}^i \veceta_k) \right] \int_{\RR_+^{m+1}} \delta(\nu_1+\cdots+\nu_{p} + r^d \sum_{i \notin \scrS} u_i - t)  \\
&\times\left[\prod_{i=1}^{p}  \e(\zeta_i \nu_{i})  \, \e^{-\alpha^d (1-\Gamma(\alpha \sum_{k=1}^i \veceta_k))\, r^{-d} \nu_i }\d \nu_i\right]  \left[ \prod_{i \notin \scrS}\e( \xi_{i}  u_{i} ) \e^{- \alpha^d u_{i}} \, \d u_{i} \right]
\end{split}
\end{equation}
where $F_{\vecgamma,\scrS}^r$ is defined by
\begin{equation}\begin{split} \label{Fdef}
F_{\vecgamma,\scrS}^r(\vecZ_\scrS,\vecH,\vecy) &= \left[\prod_{i=1}^{p-1} \bigg( \widehat W(r^d \veceta_{i+1} - \sum_{j=1}^{\kappa_i} \vecz_{ij} )\prod_{j=1}^{\kappa_i} \widehat W(\vecz_{ij}) \bigg) \right]\tilde b (\veceta_1,\vecy+\tfrac12r^d\veceta_1)\\
& \times \tilde a(- \sum_{i=1}^{p} \veceta_i,  \vecy - \sum_{i=1}^{p-1} \sum_{j=1}^{\kappa_i} (\gamma_{ij}-\gamma_{(i+1)0})  \vecz_{ij}  - \tfrac12 r^d \sum_{i=1}^p (\gamma_{i0}-\bar \gamma_{i0})\veceta_i  ).
\end{split}\end{equation}

\section{Computing the Limit $r \to 0$} \label{sec:computingthelimit}
We first separate diagonal and nondiagonal terms by writing
\begin{equation}
\scrA_{\vecgamma,\scrS}^{\alpha,r}(t) = \scrA_{\d,\vecgamma,\scrS}^{\alpha,r}(t) + \scrA_{\nd,\vecgamma,\scrS}^{\alpha,r}(t) 
\end{equation}
where $\scrA_{\d,\vecgamma,\scrS}^{r,\alpha}(t)$ is defined by restricting \eqref{simple3} to the diagonal $\vecq_{ij} = \vecq_{(i+1)0}$ for all $i=1,\dots,p-1$ and $j=1,\dots,\kappa_i$. The nondiagonal term contains the remainder of the summation.

\subsection{Nondiagonal Terms}

\begin{prop}[Upper bound on nondiagonal terms] \label{prop:convergence1} For $\alpha, t >0$, there exists a constant $C >0$ depending on $\alpha,t,W,a$ and $b$ such that
\begin{equation}
\begin{split}
\left| \scrA_{\nd,\vecgamma,\scrS}^{\alpha,r}(t)\right| &\leq C^m \, r  \, \log_2 (1+ r^{1-d} b_\scrP^{-1} ) \, \|\lambda\|_{1,\infty}^m.
\end{split}
\end{equation}
where $\|\lambda\|_{1,\infty} = \max \{ \|\lambda\|_{L^1}, \|\lambda\|_{L^\infty} \}$.
\end{prop}
The idea of the proof is simple: we note that \eqref{simple3} has the form of a Fourier transform in the $\vecz_{ij}$ variables; if we can show that this function, as well as the partial derivative $\prod_{i=1}^{p-1} \prod_{j=1}^{\kappa_i} \prod_{k=1}^d \partial_{z_{ijk}}$ of this function, is in $L^1(\RR^{d(m+1-p)})$, then the Fourier transform is bounded and decays at least linearly in each coordinate direction. This will allow us to sum over the nondiagonal terms and obtain the logarithmic bound needed. The only issue is in taking this partial derivative. Note that \eqref{simple3} contains factors of the form
\begin{equation}
\prod_{j=1}^{\kappa_i} (I-\chi)\left(\alpha(\sum_{k=1}^i \veceta_k + r^{-d} \sum_{\ell=1}^j \vecz_{i\ell}) \right).
\end{equation}
Taking the partial derivative $\prod_{j=1}^{\kappa_i} \prod_{k=1}^d \partial_{z_{ijk}}$ of this factor alone yields $ (\kappa_i!)^d$ terms by the product rule. Recall that $\kappa_i$ may be as large as $m-1$, so this would preclude us from obtaining an upper bound of the form $C^m$ as is needed. The solution to this is to first perform a carefully chosen variable substitution. Write $B_i :=\{1, \dots, \kappa_i\}$ and define
\begin{equation}
\begin{split}
\tau_{ij} &= \sum_{k=1}^j \gamma_{ik}, \\
\mu_{ij} &= \kappa_i + 1 - \sum_{k=1}^j \bar \gamma_{ik}.
\end{split}
\end{equation}
We also write $\tau_i = \tau_{i \kappa_i}$ and $\mu_i = \mu_{i\kappa_i}$ -- observe that $\mu_i = \tau_i+1$. 
\begin{lem} \label{prop2lem1} Let $M : \RR^{dm} \to \RR^{(m+1-p)}$ be defined component-wise for $i=1,\dots,p-1$ and $j=1,\dots, \kappa_i$ by
\begin{equation} \label{prop2lem1Mdef}
M(\vecq_1,\dots,\vecq_m)_{ij} = \begin{cases} \vecq_{i(j+1)}-\vecq_{ij} + \vecq_{(i+1)1} - \vecq_{(i+1)(\kappa_{i+1}+1)} & j = \sigma_i, \quad i=1,\dots,p-2 \\ \vecq_{i(j+1)}-\vecq_{ij} & \text{otherwise}\end{cases}.
\end{equation}
Then, we have that
\begin{equation} \label{prop2lem1statement}
\begin{split}
\scrA_{\nd,\vecgamma,\scrS}^{\alpha,r}(t)  &= r^{d(d-1)(p-1)}\sum_{\substack{\vecQ \in \scrP^m \\ M(\vecQ)\neq 0 }} \lambda(r^{d-1} \vecq_{1} ) \cdots \lambda(r^{d-1} \vecq_{m}) \, \widehat J_{\vecq_{1 \mu_1} \cdots \vecq_{(p-1)\mu_{p-1}} }^r(M(r^{-1} \vecq_1,\dots, r^{-1}\vecq_m) )
\end{split}
\end{equation}
where the hat denotes the usual Fourier transform and
\begin{equation} \label{prop2lem1Jdef}
\begin{split}
J_{\vecq_{1\mu_1} \cdots \vecq_{(p-1)\mu_{p-1}} }^r(\vecY_\scrS )  &= \int_{\RR^{(p+1)d}}\d\vecy \d \vecH \,\tilde a(- \sum_{i=1}^{p} \veceta_i,  \vecy_{(p-1) \sigma_{p-1}}  - \tfrac12 r^d \sum_{i=1}^p (\gamma_{i0}-\bar \gamma_{i0})\veceta_i ) \, \tilde b (\veceta_1,\vecy+\tfrac12r^d\veceta_1) \\
& \times \left[\prod_{i=1}^{p-1} \bigg( [\prod_{j=1}^{\tau_i} \widehat W(\vecy_{i(j-1)} - \vecy_{ij})] \widehat W( \vecy_{i \tau_i} - \vecy_{i \mu_i} + r^d \veceta_{i+1}) [\prod_{j=\mu_i}^{\kappa_i} \widehat W(\vecy_{i j} - \vecy_{i(j+1)})] \bigg) \right] \\
&\times  \left[ \prod_{i=1}^{p-1} \prod_{j=1}^{\kappa_i} (I-\chi)(\alpha ( \sum_{k=1}^i \veceta_k +  r^{-d} (\vecy_{i \mu_{ij} }-\vecy_{i \tau_{ij} }  )  )) \right] \left[\prod_{i=2}^{p-1}\chi(\alpha\sum_{k=1}^i \veceta_k) \right] \\
&\times  \left[ \prod_{i=1}^{p-1} \e\left(-r^{d-1} \vecq_{i \mu_i} \cdot\veceta_{i+1} \right) \right]  \int_{\RR_+^{m+1}} \delta(\nu_1+\cdots+\nu_{p} + r^d \sum_{i \notin \scrS} u_i - t)  \\
&\times\left[\prod_{i=1}^{p}  \e(\zeta_i' \nu_{i})  \, \e^{-\alpha^d (1-\Gamma(\alpha \sum_{k=1}^i \veceta_k))\, r^{-d} \nu_i }\d \nu_i\right]  \left[ \prod_{i \notin \scrS}\e( \xi_{i}'  u_{i} ) \e^{- \alpha^d u_{i}} \, \d u_{i} \right].
\end{split}
\end{equation}
\end{lem}

\begin{proof}
Permute the indices in each block, so that all those indices $s_i+j$ with $\gamma_{ij}=1$ come first, in their original order, and all those indices with $\gamma_{ij}=0$ come last, in reverse order. The equation \eqref{simple3} can be written
\begin{equation} \label{simple4}
\begin{split}
\scrA_{\vecgamma,\scrS}^{\alpha,r}(t)  &= r^{d(d-1)(p-1)}\sum_{\vecq_1, \cdots,\vecq_m \in \scrP} \lambda(r^{d-1} \vecq_1 ) \cdots \lambda(r^{d-1} \vecq_m) \int_{\RR^{(m+2)d}}\d\vecy \d \vecH \d\vecZ_\scrS\,  \\
&\times F_{\vecgamma,\scrS}^r(\vecZ_\scrS,\vecH,\vecy) \left[ \prod_{i=1}^{p-1} \prod_{j=1}^{\kappa_i} (I-\chi)(\alpha ( \sum_{k=1}^i \veceta_k + r^{-d} \sum_{\ell=1}^{\tau_{ij}} \vecz_{i\ell} + r^{-d} \sum_{\ell=\mu_{ij}}^{\kappa_i} \vecz_{i \ell} )) \right] \\
&\times  \left[ \prod_{i=1}^{p-1} \e\left(-r^{d-1} \vecq_{(i+1)0} \cdot\veceta_{i+1} - r^{-1} \sum_{j=1}^{\kappa_i} (\vecq_{ij} - \vecq_{(i+1)0}) \cdot\vecz_{ij}  \right) \right] \\
&\times\left[\prod_{i=2}^{p-1}\chi(\alpha\sum_{k=1}^i \veceta_k) \right] \int_{\RR_+^{m+1}} \delta(\nu_1+\cdots+\nu_{p} + r^d \sum_{i \notin \scrS} u_i - t)  \\
&\times\left[\prod_{i=1}^{p}  \e(\zeta_i \nu_{i})  \, \e^{-\alpha^d (1-\Gamma(\alpha \sum_{k=1}^i \veceta_k))\, r^{-d} \nu_i }\d \nu_i\right]  \left[ \prod_{i \notin \scrS}\e( \xi_{i}^*  u_{i} ) \e^{- \alpha^d u_{i}} \, \d u_{i} \right]
\end{split}
\end{equation}
where $\zeta_i$ and $F_{\gamma,\scrS}^r$ are defined as before, and
\begin{multline}
\xi_{ij}^* = \tfrac12(\|\vecy +  \sum_{k=1}^{i-1}\sum_{\ell=1}^{\kappa_k} (\bar\gamma_{k\ell} - \bar\gamma_{(k+1)0}) \vecz_{k\ell} + \sum_{\ell=\mu_{ij}}^{\kappa_i} \vecz_{i\ell} +r^{d} \sum_{k=1}^i \bar\gamma_{k0} \veceta_k\|^2 \\
 - \|\vecy -  \sum_{k=1}^{i-1}\sum_{\ell=1}^{\kappa_k} (\gamma_{k\ell} - \gamma_{(k+1)0}) \vecz_{k\ell}- \sum_{\ell=1}^{\tau_{ij}}\vecz_{i\ell} -r^{d} \sum_{k=1}^i \gamma_{k0} \veceta_k\|^2 ).
\end{multline}
We now perform the substitutions
\begin{equation}
\vecz_{ij} = \begin{cases} \vecy_{i(j-1)}-\vecy_{ij}  & j \leq \tau_{i} \\ \vecy_{i j} - \vecy_{i(j+1)} & j \geq \mu_i \end{cases}
\end{equation}
with the convention $\vecy_{i0} = \vecy_{i (\kappa_i+1)} = \vecy_{(i-1) \sigma_{i-1}}$
and $\sigma_i = \tau_i + \gamma_{(i+1)0}$. Note that
\begin{equation}
\sum_{j=1}^{\kappa_i} (\gamma_{(i+1)0}-\gamma_{ij}) \vecz_{ij} = \vecy_{i \sigma_i} - \vecy_{(i-1) \sigma_{i-1} }.
\end{equation}
We thus have
\begin{equation} \label{simple5}
\begin{split}
\scrA_{\vecgamma,\scrS}^{\alpha,r}(t)  &= r^{d(d-1)(p-1)}\sum_{\vecq_1, \cdots,\vecq_m \in \scrP} \lambda(r^{d-1} \vecq_1 ) \cdots \lambda(r^{d-1} \vecq_m) \int_{\RR^{(m+2)d}}\d\vecy \d \vecH \d\vecY_\scrS\,  \\
&\times G_{\vecgamma,\scrS}^r(\vecY_\scrS,\vecH,\vecy) \left[ \prod_{i=1}^{p-1} \prod_{j=1}^{\kappa_i} (I-\chi)(\alpha ( \sum_{k=1}^i \veceta_k +  r^{-d} (\vecy_{i \mu_{ij} }-\vecy_{i \tau_{ij} }  )  )) \right] \\
&\times  \left[ \prod_{i=1}^{p-1} \e\left(-r^{d-1} \vecq_{(i+1)0} \cdot\veceta_{i+1} + r^{-1}  \vecq_{(i+1)0} \cdot (\vecy_{i\mu_i}-\vecy_{i\tau_i} ) \right) \right] \\
&\times \left[ \prod_{i=1}^{p-1} \e \left( - r^{-1} (\vecq_{i1} \cdot \vecy_{i0} +(\vecq_{i2}-\vecq_{i1}) \cdot \vecy_{i1} + \cdots + (\vecq_{i \tau_i} - \vecq_{i (\tau_i-1)}) \cdot \vecy_{i (\tau_i-1)} - \vecq_{i \tau_i} \cdot \vecy_{i \tau_i} ) \right) \right] \\ 
&\times \left[ \prod_{i=1}^{p-1} \e \left( - r^{-1} ( \vecq_{i \mu_i} \cdot \vecy_{i \mu_i} + (\vecq_{i (\mu_i+1)} - \vecq_{i \mu_i}) \cdot \vecy_{i(\mu_i+1)} + \cdots + (\vecq_{i \kappa_i} - \vecq_{i (\kappa_i-1)}) \cdot \vecy_{i \kappa_i} -\vecq_{i\kappa_i} \cdot \vecy_{i(\kappa_i+1)} ) \right) \right] \\ 
&\times\left[\prod_{i=2}^{p-1}\chi(\alpha\sum_{k=1}^i \veceta_k) \right] \int_{\RR_+^{m+1}} \delta(\nu_1+\cdots+\nu_{p} + r^d \sum_{i \notin \scrS} u_i - t)  \\
&\times\left[\prod_{i=1}^{p}  \e(\zeta_i' \nu_{i})  \, \e^{-\alpha^d (1-\Gamma(\alpha \sum_{k=1}^i \veceta_k))\, r^{-d} \nu_i }\d \nu_i\right]  \left[ \prod_{i \notin \scrS}\e( \xi_{i}'  u_{i} ) \e^{- \alpha^d u_{i}} \, \d u_{i} \right]
\end{split}
\end{equation}
where 
\begin{equation}
\begin{split}
\zeta_i'  &=  \vecy_{(i-1)\sigma_{i-1}} \cdot (\sum_{k=1}^i \veceta_k) + \tfrac12 r^d ( \| \sum_{k=1}^{i} \bar \gamma_{s_k} \veceta_k \|^2 -\| \sum_{k=1}^{i} \gamma_{s_k} \veceta_k \|^2), \\
\xi_{ij}' &= \tfrac12(\|\vecy_{i \mu_{ij}} +r^{d} \sum_{k=1}^i \bar\gamma_{k0} \veceta_k\|^2  - \|\vecy_{i \tau_{ij} } -r^{d} \sum_{k=1}^i \gamma_{k0} \veceta_k\|^2 ),
\end{split}
\end{equation}
and
\begin{equation}\begin{split} \label{Gdef}
G_{\vecgamma,\scrS}^r(\vecY_\scrS,\vecH,\vecy) &= \left[\prod_{i=1}^{p-1} \bigg( [\prod_{j=1}^{\tau_i} \widehat W(\vecy_{i(j-1)} - \vecy_{ij})] \widehat W( \vecy_{i \tau_i} - \vecy_{i \mu_i} + r^d \veceta_{i+1}) [\prod_{j=\mu_i}^{\kappa_i} \widehat W(\vecy_{i j} - \vecy_{i(j+1)})] \bigg) \right]\\
& \times \tilde a(- \sum_{i=1}^{p} \veceta_i,  \vecy_{(p-1) \sigma_{p-1}}  - \tfrac12 r^d \sum_{i=1}^p (\gamma_{i0}-\bar \gamma_{i0})\veceta_i ) \, \tilde b (\veceta_1,\vecy+\tfrac12r^d\veceta_1).
\end{split}\end{equation}
Finally, we relabel the $\vecq_i$ indices according to the map
\begin{equation}
s_i+j \mapsto \begin{cases} s_i + j & 1 \leq j \leq \tau_{i} \\ s_i+j+1 & \mu_{i} \leq j \leq \kappa_i \\ s_i + \mu_i & j = \kappa_i+1 \end{cases}.
\end{equation}
We thus obtain
\begin{equation} \label{simple6}
\begin{split}
\scrA_{\vecgamma,\scrS}^{\alpha,r}(t)  &= r^{d(d-1)(p-1)}\sum_{\vecq_1, \cdots,\vecq_m \in \scrP} \lambda(r^{d-1} \vecq_1 ) \cdots \lambda(r^{d-1} \vecq_m) \int_{\RR^{(m+2)d}}\d\vecy \d \vecH \d\vecY_\scrS\,  \\
&\times G_{\vecgamma,\scrS}^r(\vecY_\scrS,\vecH,\vecy) \left[ \prod_{i=1}^{p-1} \prod_{j=1}^{\kappa_i} (I-\chi)(\alpha ( \sum_{k=1}^i \veceta_k +  r^{-d} (\vecy_{i \mu_{ij} }-\vecy_{i \tau_{ij} }  )  )) \right] \\
&\times  \left[ \prod_{i=1}^{p-1} \e\left(-r^{d-1} \vecq_{i \mu_i} \cdot\veceta_{i+1} -r^{-1} (\vecq_{i1} - \vecq_{i \kappa_i})\cdot \vecy_{(i-1)\sigma_{i-1}} -r^{-1} \sum_{j=1}^{\kappa_i} (\vecq_{i (j+1)}-\vecq_{i j}) \cdot \vecy_{ij}  \right) \right] \\
&\times\left[\prod_{i=2}^{p-1}\chi(\alpha\sum_{k=1}^i \veceta_k) \right] \int_{\RR_+^{m+1}} \delta(\nu_1+\cdots+\nu_{p} + r^d \sum_{i \notin \scrS} u_i - t)  \\
&\times\left[\prod_{i=1}^{p}  \e(\zeta_i' \nu_{i})  \, \e^{-\alpha^d (1-\Gamma(\alpha \sum_{k=1}^i \veceta_k))\, r^{-d} \nu_i }\d \nu_i\right]  \left[ \prod_{i \notin \scrS}\e( \xi_{i}'  u_{i} ) \e^{- \alpha^d u_{i}} \, \d u_{i} \right].
\end{split}
\end{equation}
The result then follows.
\end{proof}

\begin{lem} \label{prop2lem2} There exists a constant $C_J >0$ such that
\begin{multline}
\left| \widehat{J}_{\vecq_{1 \mu_1} \cdots \vecq_{(p-1) \mu_{p-1}} }^r( \vecxi_1,\dots,\vecxi_{m+1-p}) \right| \leq C_J^m  \frac{\langle t\rangle^{(d+1)p-1}}{(p-1)!} \, \frac{1}{\alpha^{d (m-1+p)} } \\
\| W \|_{2d}^m \, \|\chi\|_{L^1}^{p-2} \, \| \, a \, \|_{d}^* \, \| \, b \, \|_{L^1} \,\prod_{i=1}^{m+1-p} \prod_{j=1}^d \min\{ 1, \xi_{ij}^{-1}  \}
\end{multline}
where
\begin{equation}
\begin{split}
\| W \|_{N} &= \sup_{p \geq 1} \sup_{\|\vecbeta_1\|,\|\vecbeta_2\| \leq N}  \| \vecz^{\vecbeta_1} \partial_{\vecz}^{\vecbeta_2} W \|_{L^p},  \\
\| \, a \, \|_{N}^* &=\sup_{\|\vecbeta_1\|,\|\vecbeta_2\| \leq N} \int_{\RR^d} \left( \int_{\RR^d}  | \vecy^{\vecbeta_1} \partial_{\vecy}^{\vecbeta_2} \tilde a(\veceta, \vecy) |^2 \d \vecy \right)^{1/2} \d \veceta,
\end{split}
\end{equation}
and $\vecbeta_1$, $\vecbeta_2$ are multi-indices.
\end{lem}
\begin{proof}
We first prove that $J$ is in $L^1(\RR^{d(m+1-p)})$ and hence that the Fourier transform is well defined. Taking absolute values inside the integral yields
\begin{equation} \label{simpleL11}
\begin{split}
\| J_{\vecq_{1\mu_1} \cdots \vecq_{(p-1)\mu_{p-1}} }^r \|_{L^1}  &\leq \int_{\RR^{(m+2)d}}\d\vecy  \d \vecH \, \d \vecY_{\scrS}\, \left| \tilde a(- \sum_{i=1}^{p} \veceta_i,  \vecy_{(p-1)\sigma_{p-1}}  - \tfrac12 r^d \sum_{i=1}^p (\gamma_{i0}-\bar \gamma_{i0})\veceta_i  ) \, \tilde b (\veceta_1,\vecy+\tfrac12r^d\veceta_1) \right| \\
& \times \left| \prod_{i=1}^{p-1} \bigg( [\prod_{j=1}^{\tau_i} \widehat W(\vecy_{j-1} - \vecy_j)] \widehat W( \vecy_{i \tau_i} - \vecy_{i \mu_i} + r^d \veceta_{i+1}) [\prod_{j=\mu_i}^{\kappa_i} \widehat W(\vecy_{i j} - \vecy_{i(j+1)})] \bigg) \right| \\
&\times  \left[\prod_{i=2}^{p-1}\chi(\alpha\sum_{k=1}^i \veceta_k) \right] \int_{\RR_+^{m+1}} \delta(\nu_1+\cdots+\nu_{p} - t)  \d \nu_1 \cdots \d \nu_p \left[ \prod_{i \notin \scrS} \e^{- \alpha^d u_{i}} \, \d u_{i} \right].
\end{split}
\end{equation}
Integrating over $\vecnu$ and $\vecu$ yields
\begin{equation} \label{simpleL12}
\begin{split}
\| J_{\vecq_{1\mu_1} \cdots \vecq_{(p-1)\mu_{p-1}} }^r \|_{L^1}  &\leq \frac{t^{p-1}}{(p-1)!} \, \frac{1}{\alpha^{d (m+1-p)} }\int_{\RR^{(m+2)d}} \d \vecy \d \vecH \, \d \vecY_{\scrS} \, \left|\tilde b (\veceta_1,\vecy+\tfrac12r^d\veceta_1) \right|\\
&\times \left| \tilde a(- \sum_{i=1}^{p} \veceta_i,  \vecy_{(p-1)\sigma_{p-1}}  - \tfrac12 r^d \sum_{i=1}^p (\gamma_{i0}-\bar \gamma_{i0})\veceta_i  ) \right| \left[\prod_{i=2}^{p-1}\chi(\alpha\sum_{k=1}^i \veceta_k) \right]\\
& \times \left| \prod_{i=1}^{p-1} \bigg( [\prod_{j=1}^{\tau_i} \widehat W(\vecy_{j-1} - \vecy_j)] \widehat W( \vecy_{i \tau_i} - \vecy_{i \mu_i} + r^d \veceta_{i+1}) [\prod_{j=\mu_i}^{\kappa_i} \widehat W(\vecy_{i j} - \vecy_{i(j+1)})] \bigg) \right|.
\end{split}
\end{equation}
The $i^{th}$ block of $\widehat W$ factors has the form
\begin{equation}
\widehat W(\vecy_{(i-1) \sigma_{i-1}}-\vecy_{i1}) \widehat W(\vecy_{i1}-\vecy_{i2}) \cdots \widehat W(\vecy_{i\tau_i}-\vecy_{i\mu_i} + r^d \veceta_{i+1} ) \cdots \widehat W( \vecy_{i(\kappa_{i-1}-1)} - \vecy_{i \kappa_i} ) \widehat W (\vecy_{i\kappa_i}-\vecy_{(i-1) \sigma_{i-1}}).
\end{equation}
By a series of substitutions this can be written
\begin{equation}
\widehat W(\vecy_{i1}) \cdots \widehat W(\vecy_{i\kappa_i}) \widehat W(r^d \veceta_{i+1} -\vecy_{i1}-\cdots-\vecy_{i\kappa_i}). 
\end{equation}
Hence, after applying Cauchy-Schwarz to the $\vecy_{(p-1) \sigma_{p-1} }$ and $\veceta_1$ integrals we obtain
\begin{equation} \label{simpleL13}
\begin{split}
\| J_{\vecq_{1\mu_1} \cdots \vecq_{(p-1)\mu_{p-1}} }^r \|_{L^1}  &\leq \frac{t^{p-1}}{(p-1)!} \, \frac{1}{\alpha^{d (m-1)} } \| \widehat W \|_{L^\infty}^{p-1}\,\|\widehat W\|_{L^2} \, \|\widehat W \|_{L^1}^{m-p} \, \|\chi\|_{L^1}^{p-2} \, \| \, a \, \|^* \, \| \, b \, \|_{L^1}
\end{split}
\end{equation}
where
\begin{equation}
\| \, a \, \|^* = \int_{\RR^d} \left( \int_{\RR^d} |\tilde a(\veceta, \vecy) |^2 \d \vecy \right)^{1/2} \d \veceta.
\end{equation}
Next we prove that differentiating once with respect to each component of each $\vecy_{ij}$ variable yields a function which is also in $L^1$, and hence we can conclude that not only does the Fourier transform exist, it decays at least linearly in each coordinate direction.

The first step is to bound the number of terms we obtain when applying this partial derivative. The function $\tilde a$ depends only on $\vecy_{(p-1) \sigma_{p-1}}$ which appears once. The product of $\widehat W$ depends on all $\vecy_{ij}$ variables, with each one appearing either twice, if $j \neq \sigma_i$, or four times if $j=\sigma_i$. The number of terms this generates is thus bounded above by $4^{(m+1-p)d}$. The product of $(I-\chi)$ factors is more subtle. Each factor has the form
\begin{equation}
(I-\chi)\left( \alpha ( \sum_{k=1}^i + r^{-d}(\vecy_{i\mu_{ij}}-\vecy_{i\tau_{ij}}) )\right),
\end{equation}
i.e. it is a function of two $\vecy_{ij}$ variables. In passing from one factor to the next, when $\gamma_{ij}=1$ we increase the index of the second variable by one, and when $\gamma_{ij} = 0$ we decrease the index of the first variable by one. If the block consists of alternating sequences of ones and zeroes of lengths $\ell_1,\dots, \ell_n$ with $\ell_1+\cdots+\ell_n = \kappa_i$ and $n \leq \kappa_i$ then we have $n-1$ variables which appear $\ell_2+1, \dots, \ell_n+1$ times respectively, and the remaining variables appear only once. For $n\geq 2$ this yields $ ((\ell_2+1) \cdots (\ell_n+1))^d$ terms which is bounded above by $(1+ \tfrac{\kappa_i}{n})^{nd}$. This is increasing, and hence the maximum number of terms from each block is bounded above by $2^{\kappa_i d}$, and from the entire product is $2^{(m+1-p)d}$. The product of $\e(\xi_{ij}' u_{ij})$ is similar. Finally, each $\zeta_i$ depends only on $\vecy_{(i-1)\sigma_{i-1}}$. In total then, there exists a constant $C_1$ such that the number of terms is bounded above by $C_1^m$. Each time a derivative is applied to the factor $\e(\zeta_i' \nu_i)$ we obtain a multiplying factor of $\nu_i (\sum_{k=1}^i \veceta_k)$. By the compact support of $\chi$ (and the rapid decay of $\tilde a$, $\tilde b$) this is essentially bounded above by $ t \, \alpha^{-1}$. Each time a derivative is applied to the factor $\e(\xi_{ij}' u_{ij})$ we obtain a multiplying factor of $\pm u_{ij}$. There are at most $2d$ derivatives which act on each of these factors so these factors can be uniformly bounded above by e.g. $\prod_{i \notin \scrS} \langle u_i \rangle^{2d}$. Proceeding as before, there thus exists a uniform constant $C_2 >1$ such that
\begin{equation} \label{simpleL14}
\begin{split}
&\| [\prod_{i=1}^{p-1} \prod_{j=1}^{\kappa_i} \prod_{k=1}^d \frac{\partial}{\partial y_{ijk} } ] J_{\vecq_{1\mu_1} \cdots \vecq_{(p-1)\mu_{p-1}} }^r \|_{L^1}  \leq C_2^m \, \frac{\langle t\rangle^{(d+1)p-1}}{(p-1)!} \, \frac{1}{\alpha^{d (m-1+p)} } \| W \|_{2d}^m \, \|\chi\|_{L^1}^{p-2} \, \| \, a \, \|_{d}^* \, \| \, b \, \|_{L^1}.
\end{split}
\end{equation}
The result then follows.
\end{proof}

We can now prove Proposition \ref{prop:convergence1}.

\begin{proof}[Proof of Proposition \ref{prop:convergence1}]
By Lemmas \ref{prop2lem1} and \ref{prop2lem2} we have that
\begin{equation} \label{simpleL15}
\begin{split}
\left|\scrA_{\vecgamma,\scrS}^{\alpha,r}(t)\right|  &\leq C_3^m \, \frac{\langle t\rangle^{(d+1)p-1}}{(p-1)!} \, \frac{1}{\alpha^{d (m-1+p)} } \| W \|_{2d}^m \, \|\chi\|_{L^1}^{p-2} \, \| \, a \, \|_{d}^* \, \| \, b \, \|_{L^1} \\
&\times r^{d(d-1)(p-1)}\sum_{\vecQ \in \scrP^m} \lambda(r^{d-1} \vecq_{1} ) \cdots \lambda(r^{d-1} \vecq_{m})  \prod_{i=1}^{p-1} \prod_{j=1}^{\kappa_i} \prod_{k=1}^d \min \{ 1,  M(r^{-1}\vecQ)_{ijk}^{-1} \}.
\end{split}
\end{equation}
We are summing over the nondiagonal terms, so there exists an $i$ and $j$ such that $\vecq_{ij} \neq \vecq_{i \mu_i}$. In particular this implies that at least one of the $M(r^{-1} \vecQ)_{ijk}$ is nonzero. By the compact support of $\lambda$,
\begin{multline}
\sum_{\substack{\vecq\in\scrP\\ \vecq \neq \vecq'}} \lambda(r^{d-1} \vecq) \prod_{j=1}^d \min\{1, r (q_j-q_j')^{-1} \} \\
 \leq \| \lambda\|_{L^\infty}  \, \sum_{\substack{\veci \in \ZZ_{\geq 0}^d \setminus \{ 0\} \\ \| \veci \|_1 <  \log_2 (1+ r^{1-d} b_\scrP^{-1} ) }} \sum_{ \vecq\in \scrP} \prod_{j=1}^d \bm 1 \left[ 2^{i_j}-1 < \frac{|q_j-q_j'|}{ b_\scrP} < 2^{i_j+1}-1\right]   \min\{1, r (q_j-q_j')^{-1} \}.
\end{multline}
The number of points in a region of volume $V$ is bounded above by $V b_{\scrP}^{-d}$ so we conclude
\begin{equation}
\sum_{\substack{\vecq\in\scrP\\ \vecq \neq \vecq'}} \lambda(r^{d-1} \vecq) \prod_{j=1}^d \min\{1, r (q_j-q_j')^{-1} \}  \leq 2^d \, \| \lambda\|_{L^\infty} \, \sum_{\substack{\veci \in \ZZ_{\geq 0}^d \setminus \{ 0\} \\ \| \veci \|_1 <  \log_2 (1+ r^{1-d} b_\scrP^{-1} ) }}\prod_{j=1}^d\, 2^{i_j} \min\{1,\frac{r}{2^{i_j}-1}\}.
\end{equation}
In fact this can be written more simply: for $r<2$, 
\begin{equation}
2^{i} \min \{1, \frac{r}{2^{i}-1}\} = \begin{cases} 1 & i = 0 \\ r \frac{2^i}{2^i-1} & |i| > 0 \end{cases}
\end{equation}
We partition the sum into $2^d$ regions according to whether $i_j$ is zero or nonzero. The region which gives the largest contribution to the sum as $r \to 0$ is the one where all but one $i_j$ are zero. Using this upper bound we obtain
\begin{equation}
\sum_{\substack{\vecq\in\scrP\\ \vecq \neq \vecq'}} \lambda(r^{d-1} \vecq) \prod_{j=1}^d \min\{1, r (q_j-q_j')^{-1} \}  \leq 2^{2d+1} \| \lambda\|_{L^\infty} \, \, r \log_2 (1+ r^{1-d} b_\scrP^{-1} ).
\end{equation}
Hence, we may write
\begin{equation} \label{simpleL16}
\begin{split}
\left| \scrA_{\nd,\vecgamma,\scrS}^{\alpha,r}(t) \right|  &\leq 2 r  \log_2 (1+ r^{1-d} b_\scrP^{-1} ) \, \frac{\langle t \rangle^{(d+1)p-1}}{(p-1)!} \, \frac{ C_J^m}{\alpha^{d (m-1+p)} } \| W \|_{2d}^m \, \|\chi\|_{L^1}^{p-2} \, \| \, a \, \|_{d}^* \, \| \, b \, \|_{L^1} \, \\
&\times (4^d  \|\lambda\|_{L^\infty})^{m+1-p}\,  r^{d(d-1)(p-1)} \sum_{\vecq_{1\mu_1},\dots,\vecq_{(p-1)\mu_{p-1}} \in \scrP } \lambda(r^{d-1} \vecq_{1\mu_1} ) \cdots \lambda(r^{d-1} \vecq_{(p-1)\mu_{p-1}})  
\end{split}
\end{equation}
and the result follows from our assumption \eqref{assumption1}.
\end{proof}

\begin{thm}[Sum of nondiagonal terms vanishes] \label{prop:convergence2} There exists a constant $\lambda_0>0$ depending on $\alpha, t, W, a$ and $b$ such that for all $\lambda$ with $\|\lambda\|_{1,\infty}<\lambda_0$
\begin{equation} \label{prop:convergence2:equation}
\sum_{m=1}^\infty (2\pi\i)^m \sum_{\vecgamma \in\{0,1\}^m} (-1)^{\gamma_1+\cdots+\gamma_m}  \sum_{\scrS\in \Pi_m}\scrA_{\nd,\vecgamma,\scrS}^{\alpha,r}(t) = O\left(r  \log_2 (1+ r^{1-d} b_\scrP^{-1} )\right).
\end{equation}
\end{thm}
\begin{proof}
Begin from the result of Proposition \ref{prop:convergence1}. Using the fact that $|\Pi_m| = 2^{m-1}$ we see that the left hand side of \eqref{prop:convergence2:equation} is bounded above by
\begin{equation}
r  \log_2 (1+ r^{1-d} b_\scrP^{-1} ) \sum_{m=0}^\infty (8  \pi C \|\lambda\|_{1,\infty})^m \, 
\end{equation}
which converges for $\|\lambda\|_{1,\infty} < (8 \pi C)^{-1}$.

\end{proof}

\subsection{Diagonal Terms}
\begin{prop}[Convergence of diagonal terms] \label{prop:convergence3}
\begin{equation} \label{simple5}
\begin{split}
\lim_{r\to 0} \scrA_{\d,\vecgamma,\scrS}^{\alpha,r}(t) &= \int_{\RR^{(m+2)d}}\d\vecy \d \vecH \d\vecZ_\scrS  \, \widehat \lambda_{\kappa_1} (\veceta_2) \cdots \widehat \lambda_{\kappa_{p-1}}(\veceta_p) \, F_{\gamma,\scrS}(\vecZ_\scrS,\vecH)   \\
&\times \Bigg[  \prod_{i\notin\scrS} \int_{\RR_+} \e( \xi_{i}^0 u ) \e^{- \alpha^d u} \, \d u \Bigg] \int_{\RR_+^{p}}   \d \vecnu \, \delta(\nu_1+\cdots+\nu_{p} - t) \\
&\times \left[\prod_{i=1}^{p}  \e( (\vecy -  \sum_{k=1}^{i-1} \sum_{\ell=1}^{\kappa_k} (\gamma_{k\ell} - \gamma_{(k+1)0} ) \vecz_{k\ell} )\cdot (\sum_{k=1}^i \veceta_k) \, \nu_{i}) \, \bm 1 [ \Gamma(\alpha \sum_{k=1}^i \veceta_k)=1 ]\right].
\end{split}
\end{equation}
where $ \widehat \lambda_\kappa (\veceta) = \int_{\RR^d} [ \lambda(\vecx)]^\kappa \, \e(-\vecx\cdot\veceta) \, \d \vecx$,
\begin{equation}\begin{split} \label{F0def}
F_{\gamma,\scrS}(\vecZ_\scrS,\vecH) &= \left[\prod_{i=1}^{p-1} \widehat W(- \sum_{j=1}^{\kappa_i} \vecz_{ij} )\prod_{j=1}^{\kappa_i} \widehat W(\vecz_{ij})\right] \\
& \times \tilde a(- \sum_{i=1}^{p} \veceta_i,  \vecy - \sum_{i=1}^{p-1} \sum_{j=1}^{\kappa_i} (\gamma_{ij}-\gamma_{(i+1)0})  \vecz_{ij} ) \, \tilde b (\veceta_1,\vecy),
\end{split}\end{equation}
and for $i=1,\dots,p-1$ and $j=1,\dots,\kappa_i$
\begin{multline}
\xi_{ij}^0 = \tfrac12(\|\vecy +  \sum_{k=1}^{i-1}\sum_{\ell=1}^{\kappa_i} (\bar\gamma_{k\ell} - \bar\gamma_{(k+1)0}) \vecz_{k\ell} + \sum_{\ell=1}^j \bar\gamma_{i\ell} \vecz_{i\ell}\|^2 \\
- \|\vecy -  \sum_{k=1}^{i-1}\sum_{\ell=1}^{\kappa_i} (\gamma_{k\ell} - \gamma_{(k+1)0}) \vecz_{k\ell}- \sum_{\ell=1}^j \gamma_{i\ell} \vecz_{i \ell}\|^2 ) .
\end{multline}
\end{prop}
\begin{proof}
From \eqref{simple3} and the definition of the diagonal terms we have
\begin{equation}
\begin{split}
\scrA_{\d,\vecgamma,\scrS}^{\alpha,r}(t)  &= r^{d(d-1)(p-1)}\sum_{\vecq_{s_2}, \cdots,\vecq_{s_p} \in \scrP} \lambda(r^{d-1} \vecq_{s_2} )^{\kappa_1} \cdots \lambda(r^{d-1} \vecq_{s_p})^{\kappa_{p-1}} \int_{\RR^{(m+2)d}}\d\vecy \d \vecH \d\vecZ_\scrS\,  \\
&\times F_{\gamma,\scrS}^r(\vecZ_\scrS,\vecH,\vecy) \left[ \prod_{i=1}^{p-1} \prod_{j=1}^{\kappa_i} (I-\chi)(\alpha(\sum_{k=1}^i \veceta_k + r^{-d} \sum_{\ell=1}^j \vecz_{i\ell})) \right] \\
&\times  \left[ \prod_{i=1}^{p-1} \e\left(-r^{d-1} \vecq_{(i+1)0} \cdot\veceta_{i+1}  \right) \right] \\
&\times\left[\prod_{i=2}^{p-1}\chi(\alpha\sum_{k=1}^i \veceta_k) \right] \int_{\RR_+^{m+1}} \delta(\nu_1+\cdots+\nu_{p} + r^d \sum_{i \notin \scrS} u_i - t)  \\
&\times\left[\prod_{i=1}^{p}  \e(\zeta_i \nu_{s_i})  \, \e^{-\alpha^d (1-\Gamma(\alpha \sum_{k=1}^i \veceta_k))\, r^{-d} \nu_i }\d \nu_i\right]  \left[ \prod_{i \notin \scrS}\e( \xi_{i}  u_{i} ) \e^{- \alpha^d u_{i}} \, \d u_{i} \right]
\end{split}
\end{equation}
By assumption \eqref{assumption1}, we have that
\begin{equation}
r^{d(d-1)}\sum_{\vecq \in \scrP} [\lambda(r^{d-1} \vecq )]^\kappa \e \left( - r^{d-1} \vecq \cdot \veceta \right) =  \widehat \lambda_{\kappa}(\veceta) + O(r^{(d-1) c_\scrP} \|\veceta\| ).
\end{equation}
We thus obtain the upper bound
\begin{equation}
\begin{split}
\left| \scrA_{\d,\vecgamma,\scrS}^{\alpha,r}(t) \right|  &\leq \frac{t^{p-1}}{(p-1)!} \, \frac{1}{\alpha^{d(m+1-p)}} \int_{\RR^{(m+2)d}}\d\vecy \d \vecH \d\vecZ_\scrS\, \left| F_{\gamma,\scrS}^r(\vecZ_\scrS,\vecH,\vecy) \right|\\ 
&\times\left[ \prod_{i=2}^p \left( \widehat\lambda_{\kappa_{i-1}}(\veceta_{i} ) + O( r^{(d-1) c_\scrP} \|\veceta_i\| ) \right) \chi(\alpha \sum_{k=1}^i \veceta_k) \right] .
\end{split}
\end{equation}
Since $\chi$ is compactly supported, this integral converges and we can apply dominated convergence. The result then follows by taking the pointwise limit $r\to 0$, using the fact that for all $c>0$
\begin{equation}
\lim_{r\to 0} \e^{- c (1-\Gamma(\vecz) ) r^{-d} } = \bm 1[ \Gamma(\vecz)=1]
\end{equation}
and that for all $\vecz \neq 0$ we have
\begin{equation}
\lim_{r \to 0} \, \chi(\alpha(\sum_{k=1}^i \veceta_k + r^{-d}\vecz )) = 0.
\end{equation}
\end{proof}
\begin{thm}[Sum of diagonal terms converges] \label{prop:convergence4} For $\alpha, t>0$ fixed, there exists a constant $\lambda_0>0$ such that for all $\lambda$ with $\|\lambda\|_{1,\infty}<\lambda_0$, the series
\begin{equation} \label{prop:convergence4:equation}
\sum_{m=1}^\infty (2\pi\i)^m \sum_{\vecgamma \in\{0,1\}^m} (-1)^{\gamma_1+\cdots+\gamma_m}  \sum_{\scrS\in \Pi_m}\scrA_{\d,\vecgamma,\scrS}^{\alpha,r}(t)
\end{equation}
is absolutely convergent, uniformly as $r \to 0$.
\end{thm}
\begin{proof}
Proceeding as in the Proof of Proposition \ref{prop:convergence3}, we may write
\begin{equation}
\begin{split}
\left| \scrA_{\d,\vecgamma,\scrS}^{\alpha,r}(t) \right|  &\leq \frac{t^{p-1}}{(p-1)!} \, \frac{1}{\alpha^{d(m+1-p)}} \int_{\RR^{(m+2)d}}\d\vecy \d \vecH \d\vecZ_\scrS\, \left| F_{\gamma,\scrS}^r(\vecZ_\scrS,\vecH,\vecy) \right|\\ 
&\times\left[ \prod_{i=2}^p \left( \widehat \lambda_{\kappa_{i-1}}(\veceta_{i} ) + O( r^{(d-1) c_\scrP} \|\veceta_i\| ) \right) \chi(\alpha \sum_{k=1}^i \veceta_k) \right] .
\end{split}
\end{equation}
Since $\chi$ is compactly supported and $\tilde a$ and $\tilde b$ are rapidly decaying, the integral over $\vecH$ converges. By the definition of $F_{\vecgamma,\scrS}^r$ there exists a constant $C>0$ such that $ \left|\scrA_{\d,\vecgamma,\scrS}^{\alpha,r}(t)\right|  < C^{m+1}$. Equation \eqref{prop:convergence4:equation} can thus be bounded above by
\begin{equation}
\sum_{m=0}^\infty (8 \pi \|\lambda\|_{1,\infty} C)^m
\end{equation}
which converges for $\|\lambda\|_{1,\infty} < (8 \pi C)^{-1}$.

\end{proof}

\subsection{The zero-damping limit}

\begin{prop}[Convergence of diagonal terms] \label{prop:convergence5}
\begin{equation}
\lim_{\alpha \to 0}\lim_{r\to 0} \scrA_{\d,\vecgamma,\scrS}^{\alpha,r}(t) = \int_{\RR^{2d}} \d \vecx \d\vecy \, f_{\vecgamma,\scrS}(t,\vecx,\vecy) \, b(\vecx,\vecy)
\end{equation}
where
\begin{equation*} \label{simple8}
\begin{split}
f_{\vecgamma,\scrS}(t,\vecx,\vecy) &= \int_{\RR^{(m+1-p)d}} \d\vecZ_\scrS \int_{\RR_+^{p}}   \d \vecnu \, \delta(\nu_1+\cdots+\nu_{p} - t) \\
&\times \left[ \prod_{i=2}^{p} \left[\lambda \left(\vecx- \sum_{j=1}^{i-1}\bigg( \vecy- \sum_{k=1}^{j-1} \sum_{\ell=1}^{\kappa_k} (\gamma_{k\ell}-\gamma_{(k+1)0}) \vecz_{k\ell} \bigg) \nu_j\right) \right]^{\kappa_{i-1}}\right] \\
&\times  \left[\prod_{i=1}^{p-1} \widehat W(- \sum_{j=1}^{\kappa_i} \vecz_{ij} )\prod_{j=1}^{\kappa_i} \widehat W(\vecz_{ij})\right] \left[  \prod_{i=1}^{p-1}\prod_{j=1}^{\kappa_i} \int_{\RR_+} \e( \xi_{ij}^0 u_{ij} ) \, \d u_{ij} \right]  \\
&\times a(\vecx- t \vecy + \sum_{i=1}^{p} \sum_{k=1}^{i-1} \sum_{\ell=1}^{\kappa_k} (\gamma_{k\ell} - \gamma_{(k+1)0} ) \vecz_{k\ell} \, \nu_{i},  \vecy - \sum_{i=1}^{p-1} \sum_{j=1}^{\kappa_i} (\gamma_{ij}-\gamma_{(i+1)0})  \vecz_{ij} ).
\end{split}
\end{equation*}
\end{prop}
\begin{proof} We begin from the statement of Proposition \ref{prop:convergence3}, and claim that the $\vecu$ integral converges uniformly for $\alpha \geq 0$. Using the same substitutions as in the proof of Proposition \ref{prop:convergence1} we can write
\begin{equation}
\begin{split}
\lim_{r\to 0} \scrA_{\d,\vecgamma,\scrS}^{\alpha,r}(t) &= \int_{\RR^{(m+2)d}}\d\vecy \d \vecH \d\vecZ_\scrS  \, \widehat \lambda_{\kappa_1} (\veceta_2) \cdots \widehat \lambda_{\kappa_{p-1}}(\veceta_p) \, G_{\gamma,\scrS}(\vecY_\scrS,\vecH)   \\
&\times \Bigg[  \prod_{i=1}^{p-1} \prod_{j=1}^{\kappa_i} \int_{\RR_+} \e( \tfrac12 (\|\vecy_{i\mu_{ij}}\|^2 - \|\vecy_{i\tau_{ij} } \|^2 ) u ) \e^{- \alpha^d u} \, \d u \Bigg] \int_{\RR_+^{p}}   \d \vecnu \, \delta(\nu_1+\cdots+\nu_{p} - t) \\
&\times \left[\prod_{i=1}^{p}  \e( \vecy_{(i-1) \sigma_{i-1}} \cdot (\sum_{k=1}^i \veceta_k) \, \nu_{i}) \, \bm 1 [ \Gamma(\alpha \sum_{k=1}^i \veceta_k)=1 ]\right]
\end{split}
\end{equation}
where
\begin{equation}\begin{split} 
G_{\vecgamma,\scrS}(\vecY_\scrS,\vecH,\vecy) &= \left[\prod_{i=1}^{p-1} \prod_{j=1}^{\kappa_i+1} \widehat W(\vecy_{i(j-1)} - \vecy_{ij})  \right] \, \tilde a(- \sum_{i=1}^{p} \veceta_i,  \vecy_{(p-1) \sigma_{p-1}} ) \, \tilde b (\veceta_1,\vecy)
\end{split}\end{equation}
and we have the convention $\vecy_{i0} = \vecy_{i (\kappa_i+1)} = \vecy_{(i-1) \sigma_{i-1}}$. Considering only the $\vecy_{ij}$ integration this has the form
\begin{equation} \label{toymodel}
\int_{\RR^{(m+1-p)d}} g(\vecY_\scrS) \, \Bigg[  \prod_{i=1}^{p-1} \prod_{j=1}^{\kappa_i} \int_{\RR_+} \e( \tfrac12 (\|\vecy_{i\mu_{ij}}\|^2 - \|\vecy_{i\tau_{ij} } \|^2 ) u ) \e^{- \alpha^d u} \, \d u \Bigg] \d \vecY_{\scrS}
\end{equation}
where $g$ is Schwartz class uniformly in $\alpha$. Consider just the first block $\{\gamma_1,\dots,\gamma_{\kappa_1} \}$, and suppose that it consists of sub-blocks of $\ell_1$ ones, followed by $\ell_2$ zeroes, followed by $\ell_3$ ones, and so on. If there are $2k$ of these sub-blocks in total then
\begin{multline} \label{toymodelexpandedproduct}
\prod_{j=1}^{\kappa_1} \int_{\RR_+} \e( \tfrac12 (\|\vecy_{i\mu_{ij}}\|^2 - \|\vecy_{i\tau_{ij} } \|^2 ) u ) \e^{- \alpha^d u} \, \d u  \\
= \int_{\RR_+^{\kappa_1} } \d u_1 \cdots \d u_{\kappa_1} \left( \prod_{i=1}^{\ell_1}\e( \tfrac12 (\|\vecy\|^2 - \|\vecy_{i} \|^2 ) u_i ) \e^{- \alpha^d u_i} \right) \left( \prod_{i= \kappa_1+1-\ell_2}^{\kappa_1} \e( \tfrac12 (\|\vecy_{i}\|^2 - \|\vecy_{\ell_1} \|^2 ) u_i ) \e^{- \alpha^d u_i} \right)   \\
\times \cdots  \times\left( \prod_{i=\ell_1+\cdots+\ell_{2k-3}+1}^{\ell_1+\cdots+\ell_{2k-1}} \e( \tfrac12 (\|\vecy_{\kappa_1+1-\ell_2-\cdots-\ell_{2k-2}} \|^2 - \|\vecy_{i} \|^2 ) u_i ) \e^{- \alpha^d u_i} \right)  \\
\times \left( \prod_{i=\kappa_1+1-\ell_2-\cdots-\ell_{2k}}^{\kappa_1-\ell_2-\cdots-\ell_{2k-2}} \e( \tfrac12 (\|\vecy_{i}\|^2 - \|\vecy_{\ell_1+\cdots+\ell_{2k-1}} \|^2 ) u_i) \e^{- \alpha^d u_i} \right).
\end{multline}
Let $g_0 \in \scrS(\RR^d)$ be Schwartz class, then by stationary phase one obtains
\begin{equation}
\int_{\RR^d} g_0(\vecy) \e( \tfrac12 \|\vecy\|^2 \, s) \d \vecy \ll \langle s\rangle^{-d/2}
\end{equation}
where $ A \ll B$ means there exists a constant $c$ such that $A < c B$. Most of the $\vecy_i$ appear in only one factor, as $\e(\pm \tfrac12 \|\vecy_i\|^2 u_i)$, which after integrating over $\vecy_i$ against the Schwartz function $g$ gives a factor $\langle u_i \rangle^{-d/2}$. If $i = \ell_1+\cdots+ \ell_{2j-1}$ or $\kappa_1-\ell_2-\cdots-\ell_{2j}$ for some $j$ then it appears with a more complicated coefficient. For example, $\vecy_{\ell_1}$ appears as the exponential factor
\begin{equation}
\e( -\tfrac12 \|\vecy_{\ell_1} \|^2 ( u_{\ell_1} + u_{\kappa_1+1-\ell_2} +\cdots+ u_{\kappa_i} ) ).
\end{equation}
After integrating over $\vecy_{\ell_1}$ this yields a factor of $\langle u_{\ell_1}+ u_{\kappa_1+1-\ell_2} +\cdots+ u_{\kappa_i} \rangle^{-d/2}$, but since all the $u_i$ are nonnegative, this can be bounded above by $\langle u_{\ell_1} \rangle^{-d/2}$. In other words, we have that
\begin{multline} \label{toymodel2}
\int_{\RR^{(m+1-p)d}} g(\vecY_\scrS) \, \Bigg[  \prod_{i=1}^{p-1} \prod_{j=1}^{\kappa_i} \int_{\RR_+} \e( \tfrac12 (\|\vecy_{i\mu_{ij}}\|^2 - \|\vecy_{i\tau_{ij} } \|^2 ) u ) \e^{- \alpha^d u} \, \d u \Bigg] \d \vecY_{\scrS} \\
\ll \int_{\RR_+^{ m+1-p} } \prod_{i=1}^{p-1} \prod_{j=1}^{\kappa_i}   \langle u_{ij} \rangle^{-d/2} \d u_{ij}
\end{multline}
uniformly for all $\alpha \geq 0$. These integrals converge for all $d\geq 3$, and the result then follows by integrating over $\vecH$ and setting $\alpha = 0$ in \eqref{simple5}.
\end{proof}

\begin{thm}[Convergence of the full series] \label{thm:convergenceofseries}
There exists a constant $\lambda_0>0$ depending on $\alpha, t, W, a$ and $b$ such that for all $\lambda$ with $\|\lambda\|_{1,\infty}<\lambda_0$
\begin{equation}
\lim_{\alpha \to 0}\lim_{r\to 0} \Tr(\rho_{r^{1-d}t} \, \Op_{r,h}(b) ) = \int_{\RR^{2d}} \d \vecx \d\vecy \, f(t,\vecx,\vecy) \, b(\vecx,\vecy)
\end{equation}
where
\begin{equation} \label{weaklimit}
f(t,\vecx,\vecy) = a(\vecx-t\vecy,\vecy) + \sum_{m=1}^\infty (2 \pi \i)^m \sum_{\vecgamma \in \{0,1\}^m} (-1)^{\gamma_1+\cdots+\gamma_m} \sum_{\scrS \in \Pi_m}f_{\vecgamma,\scrS}(t,\vecx,\vecy).
\end{equation}
\end{thm}
\begin{proof}
We begin from the definition
\begin{equation}
\Tr (\rho_{r^{1-d} t} \Op_{r,h}(b) ) = \sum_{m=0}^\infty (2 \pi \i )^m \, \scrA_m^{\alpha,r}(t).
\end{equation}
The $m=0$ term converges by Proposition \ref{prop:mequalszero}. Separating the remaining terms into diagonal and nondiagonal parts gives
\begin{multline}
\sum_{m=1}^\infty (2 \pi \i )^m \, \scrA_m^{\alpha,r}(t)= \sum_{m=1}^\infty (2 \pi \i )^m \,  \sum_{\vecgamma \in\{0,1\}^m} (-1)^{\gamma_1+\cdots+\gamma_m}  \sum_{\scrS\in \Pi_m}\scrA_{\d,\vecgamma,\scrS}^{\alpha,r}(t) \\
+\sum_{m=1}^\infty (2 \pi \i)^m \,  \sum_{\vecgamma \in\{0,1\}^m} (-1)^{\gamma_1+\cdots+\gamma_m}  \sum_{\scrS\in \Pi_m}\scrA_{\nd,\vecgamma,\scrS}^{\alpha,r}(t)
\end{multline}
Applying Theorems \ref{prop:convergence2} and \ref{prop:convergence4} tells us that for $\lambda$ small enough, the first term on the right hand side converges, and that the second vanishes in the limit $r \to 0$. Following the proof of Proposition \ref{prop:convergence5} there exists a constant $C>0$ such that
\begin{equation}
\left| \lim_{r\to 0} \scrA_{\d,\vecgamma,\scrS}^{\alpha,r}(t) \right| < C^m
\end{equation}
uniformly for all $\alpha \geq 0$. The series thus converges uniformly for $\|\lambda\|_{1,\infty}< (8 \pi C)^{-1}$ and the result follows from Proposition \ref{prop:convergence5}.
\end{proof}

\section{Extracting the  Linear Boltzmann Equation}
We are now ready to prove Theorem \ref{theorem}, namely that the weak limit, $f(t,\vecx,\vecy)$, in Theorem \ref{thm:convergenceofseries} coincides with a solution of the linear Boltzmann equation. We first show that it satisfies an auxiliary transport equation.
\begin{prop}
The expression $f(t,\vecx,\vecy)$ in \eqref{weaklimit} satisfies
\begin{equation}\label{transport1}
\begin{split}
(\partial_t &+\vecy\cdot\nabla_\vecx) f(t,\vecx,\vecy)\\
&= 2 \Re \bigg\{ \sum_{n=2}^\infty (-2 \pi \i \lambda (\vecx) )^{n} \sum_{\vecgamma \in \{0,1\}^{n-1}} (-1)^{\gamma_1+\cdots+\gamma_{n-1}} \\
&\times \int_{\RR^{d(n-1)} } \d \vecz_1 \cdots \d \vecz_{n-1}\, \widehat W(-\vecz_1) \cdots \widehat W(-\vecz_{n-1}) \widehat W(\vecz_1+\cdots + \vecz_{n-1})  \\
&\times\bigg[\prod_{i=1}^{n-1} \int_{\RR_+}\e( \tfrac12( \|\vecy- \sum_{j=1}^i \gamma_j \vecz_j \|^2 - \| \vecy  + \sum_{j=1}^i \bar\gamma_j \vecz_j \|^2)\,u) \d u \bigg]  f(t,\vecx,\vecy- \sum_{i=1}^{n-1} \gamma_i \vecz_i) \,  \bigg\}
\end{split}
\end{equation}
\end{prop}
\begin{proof}
Note that every $\scrS \in \Pi_m$ can be `decomposed' into two pieces: if $\scrS = \{0, n, \dots, m\}$ we decompose it into the pieces $\{0,n\}$ and $\{n, \dots , m\}$. Through this decomposition the function $f(t,\vecx,\vecy)$ can be written recursively as
\begin{equation}\label{weaklimit2}
\begin{split}
f(t,\vecx,\vecy) &= a(\vecx-t\vecy,\vecy) \\
&+  \sum_{n=1}^\infty  (2 \pi \i )^n \sum_{\vecgamma \in \{0,1\}^n} (-1)^{\gamma_1+\cdots+\gamma_n}  \int_0^t \d \nu \, \left[\lambda\bigg(\vecx - (t-\nu)\vecy  \bigg)\right]^n\\
&\times \int_{\RR^{(n-1)d}} \d \vecz_1\cdots \d \vecz_{n-1}  \widehat W(\vecz_1) \cdots \widehat W(\vecz_{n-1}) \widehat W(-\vecz_1-\cdots-\vecz_{n-1} ) \\
&\times \bigg[\prod_{i=1}^{n-1} \int_{\RR_+}\e( \tfrac12( \|\vecy+ \sum_{j=1}^i \bar\gamma_j \vecz_j \|^2 - \| \vecy  - \sum_{j=1}^i \gamma_j \vecz_j \|^2)\,u) \d u \bigg]\\
&\times f(\nu, \vecx- (t-\nu )\vecy, \vecy - (\gamma_1-\gamma_n)\vecz_1-\cdots-(\gamma_{n-1}-\gamma_n)\vecz_{n-1} ) .
\end{split}
\end{equation} 
The $n=1$ term vanishes -- $\gamma_1 = 1$ and $\gamma_1=0$ yield the same expression with opposite signs. Applying the operator $(\partial_t+\vecy\cdot\nabla_\vecx)$ to both sides yields
\begin{equation}\label{weaklimit3}
\begin{split}
(\partial_t &+ \vecy \cdot \nabla_\vecx)f(t,\vecx,\vecy) \\
&=\sum_{n=2}^\infty  (2 \pi \i \lambda(\vecx))^n \sum_{\vecgamma \in \{0,1\}^n} (-1)^{\gamma_1+\cdots+\gamma_n}  \\
&\times \int_{\RR^{(n-1)d}} \d \vecz_1\cdots \d \vecz_{n-1}  \widehat W(\vecz_1) \cdots \widehat W(\vecz_{n-1}) \widehat W(-\vecz_1-\cdots-\vecz_{n-1} ) \\
&\times\bigg[\prod_{i=1}^{n-1} \int_{\RR_+}\e( \tfrac12( \|\vecy+ \sum_{j=1}^i \bar\gamma_j \vecz_j \|^2 - \| \vecy  - \sum_{j=1}^i \gamma_j \vecz_j \|^2)\,u) \d u \bigg]\\
&\times f(t, \vecx, \vecy - (\gamma_1-\gamma_n)\vecz_1-\cdots-(\gamma_{n-1}-\gamma_n)\vecz_{n-1} )
 \end{split}
\end{equation} 
By summing over $\gamma_n$ we obtain
\begin{equation}
\begin{split}
(\partial_t &+ \vecy \cdot \nabla_\vecx)f(t,\vecx,\vecy)\\
&=\sum_{n=2}^\infty  (2 \pi \i \lambda(\vecx))^n \sum_{\vecgamma \in \{0,1\}^{n-1}} (-1)^{\gamma_1+\cdots+\gamma_{n-1}}  \\
&\times \int_{\RR^{(n-1)d}} \d \vecz_1\cdots \d \vecz_{n-1}  \widehat W(\vecz_1) \cdots \widehat W(\vecz_{n-1}) \widehat W(-\vecz_1-\cdots-\vecz_{n-1} ) \\
&\times\bigg[\prod_{i=1}^{n-1} \int_{\RR_+}\e( \tfrac12( \|\vecy+ \sum_{j=1}^i \bar\gamma_j \vecz_j \|^2 - \| \vecy  - \sum_{j=1}^i \gamma_j \vecz_j \|^2)\,u) \d u \bigg]\\
&\times \left( f(t, \vecx, \vecy - \sum_{i=1}^{n-1}\gamma_i\vecz_i) - f(t, \vecx, \vecy + \sum_{i=1}^{n-1}\bar\gamma_i\vecz_i ) \right).
\end{split}
\end{equation}
For the second term, we replace $\gamma_i$ by $\bar\gamma_i$ and make the variable substitutions $\vecz_i \to - \vecz_i$. This allows us to combine the two terms and the result follows.

\end{proof}

\begin{proof}[Proof of Theorem \ref{theorem}]
Define the distribution
\begin{equation}
\Delta(\vecn,\vecp) := \int_0^{\infty} \exp\{ \i (\|\vecn\|^2-\|\vecp\|^2) \, s \} \, \d s,
\end{equation}
and put $\widehat V(\vecy) = -2 \widehat W(-\vecy)$. Then, \eqref{transport1} can be written
\begin{equation}\label{transport2}
\begin{split}
(\partial_t &+\vecy\cdot\nabla_\vecx) f(t,\vecx,\vecy)\\
&= 2 \pi \Re \bigg\{ \sum_{n=2}^\infty\lambda (\vecx)^{n} \sum_{\vecgamma \in \{0,1\}^{n-1}} (-1)^{\gamma_1+\cdots+\gamma_{n-1}} \\
&\times \int_{\RR^{d(n-1)} } \d \vecz_1 \cdots \d \vecz_{n-1}\, [ \i \widehat V(\vecz_1)] \cdots [\i \widehat V(\vecz_{n-1})] [\i \widehat V(-\vecz_1-\cdots - \vecz_{n-1})]  \\
&\times\bigg[\prod_{i=1}^{n-1}\Delta(\vecy - \sum_{j=1}^i \gamma_j \vecz_j,  \vecy  + \sum_{j=1}^i \bar\gamma_j \vecz_j ) \bigg]  f(t,\vecx,\vecy- \sum_{i=1}^{n-1} \gamma_i \vecz_i) \,  \bigg\}.
\end{split}
\end{equation} 
i.e.
$$(\partial_t +\vecy\cdot\nabla_\vecx) f(t,\vecx,\vecy) =\pi \sum_{\ell=1}^\infty \lambda(\vecx)^{\ell+1} \scrQ_{\ell}(g_\vecx)(t,\vecy)$$ with $\scrQ_\ell$ as in \cite[Eq (2.7)]{Castella_LD2} and $g_\vecx: (t,\vecy) \mapsto f(t,\vecx,\vecy)$. In view of \cite[Lemma 3 \& Theorem 2]{Castella_LD2} (Recall that $W$ and the initial data $a$ are both Schwartz, so certainly satisfy the weaker regularity assumptions made in \cite{Castella_LD2}) we obtain
\begin{equation}\label{transport3}
\begin{split}
(\partial_t &+\vecy\cdot\nabla_\vecx) f(t,\vecx,\vecy) = \pi \int_{\RR^d}  [ \Sigma^{\mathrm{ld}}(\vecy,\vecy') f(t,\vecx,\vecy')- \Sigma^{\mathrm{ld}}(\vecy',\vecy) f(t,\vecx,\vecy)] \, \d \vecy'
\end{split}
\end{equation}
where
\begin{equation}
\Sigma^{\mathrm{ld}}(\vecy,\vecy') = 2 \pi \delta(\|\vecy\|^2-\|\vecy'\|^2) |\scrT(\vecy',\vecy)|^2
\end{equation}
and (see \cite[(2.5)]{Castella_LD2})
\begin{equation}
\begin{split}
\scrT(\vecy',\vecy) &= \lambda(\vecx) \widehat V(\vecy'-\vecy) \\
&-\i \sum_{\ell=1}^\infty (\i \lambda(\vecx))^{\ell+1} \int_{\RR^{\ell d}} \widehat V(\vecy'-\veck_1) \cdots \widehat V(\veck_\ell-\vecy) \Delta(\vecy,\veck_1) \cdots \Delta(\vecy,\veck_\ell) \, \d \veck_1\cdots \d \veck_\ell \\
&= -2 \lambda(\vecx) \widehat W(\vecy-\vecy') \\
&- 2 \sum_{\ell=1}^{\infty} \lambda(\vecx)^{\ell+1} (-2  \pi \i)^{\ell} \int_{\RR^{\ell d}} \widehat W(\vecy-\veck_1) \cdots \widehat W(\veck_\ell-\vecy') \\ 
&\hspace{4cm}\times \left[\int_0^\infty \prod_{i=1}^{\ell} \e( \tfrac12 (\|\vecy\|^2-\|\veck_i\|^2) \, u) \, \d u \right] \, \d \veck_1 \cdots \d \veck_\ell.
\end{split}
\end{equation}
Hence, $\scrT(\vecy',\vecy) = -2 T(\vecy,\vecy')$ and
\begin{equation}
\pi \Sigma^{\mathrm{ld}}(\vecy,\vecy')  = 8 \pi^2 \delta(\|\vecy\|^2 - \|\vecy'\|^2) |T_{\lambda(x)}(\vecy,\vecy')|^2
\end{equation}
with $T_{\mu}$ as in \eqref{Texplicit}.

\end{proof}

\bigskip

\noindent Remark: The relation $\scrT(\vecy',\vecy) = - 2 T(\vecy,\vecy')$ is due to a number of minor differences between the present set-up and Castella's work \cite{Castella_LD,Castella_LD2}: (i) the Fourier transforms are normalised differently, (ii) the Schr\"odinger operator is normalised differently, (iii) the initial von Neumann equation \eqref{heisenberg} has $\vecy$ and $\vecy'$ interchanged.

\end{document}